\documentclass[ejs]{imsart}

\RequirePackage[numbers]{natbib}
\RequirePackage[colorlinks,citecolor=blue,urlcolor=blue]{hyperref}
\usepackage{amsmath,amssymb,epsfig,amsthm,bm,xcolor}
\usepackage{enumerate,mathrsfs,graphics,caption,subfig,booktabs,verbatim} 
\usepackage{dsfont}
\usepackage{array,multirow}
\newcommand{\STAB}[1]{\begin{tabular}{@{}c@{}}#1\end{tabular}}
\usepackage{graphicx}
\theoremstyle{definition}
\newtheorem{de}{Definition}

\theoremstyle{plain}
\newtheorem{theo}[de]{Theorem}

\theoremstyle{remark}

\newcommand{\R}{\mathbb{R}}

\newcommand{\p}{\mathbb{P}}

\newcommand{\E}{\mathbb{E}}

\newcommand{\G}{\mathcal{G}}

\newcommand{\X}{\mathcal{X}}

\newcommand{\1}{\mathds{1}}

\newcommand{\dd}{\mathrm{d}}

\doi{10.1214/154957804100000000}
\pubyear{0000}
\volume{0}
\firstpage{1}
\lastpage{1}
\startlocaldefs
\endlocaldefs

\begin{document}

\begin{frontmatter}

% "Title of the Paper"
\title{A 2-step estimation procedure for semiparametric mixture cure models}
\runtitle{2-step estimation for mixture cure models}

% indicate corresponding author with \corref{}
% \author{\fnms{John} \snm{Smith}\thanksref{t1}\corref{}\ead[label=e1]{smith@foo.com}\ead[label=e2,url]{www.foo.com}}
% \thankstext{t1}{Thanks to somebody} 
% \address{line 1\\ line 2\\ \printead{e1}\\ \printead{e2}}

 \author[A]{\fnms{Eni} \snm{Musta}\corref{}\ead[label=e1]{e.musta@uva.nl}},
\author[B]{\fnms{Valentin} \snm{Patilea}\ead[label=e2]{valentin.patilea@ensai.fr}}
\and
\author[C]{\fnms{Ingrid} \snm{Van Keilegom}\ead[label=e3]{ingrid.vankeilegom@kuleuven.be}}

\address[A]{Korteweg de Vries Institute for Mathematics, University of Amsterdam, Netherlands, \printead{e1}}
\address[B]{CREST, Ensai, France, \printead{e2}}
\address[C]{ORSTAT, KU Leuven, Belgium, \printead{e3}}

\runauthor{E. Musta, V. Patilea \& I. Van Keilegom }

%\author{\fnms{???} \snm{???}\ead[label=e1]{???}}
%\address{\printead{e1}}
%\and
%\author{\fnms{???} \snm{???}\ead[label=e2]{???}}
%\address{\printead{e2}}

\begin{abstract}
Cure models have been developed as an alternative modelling approach to conventional survival analysis in order to account for the presence  of cured subjects that will never experience the event of interest. Mixture cure models, which model separately the cure probability and the survival of uncured subjects depending on a set of covariates, are  particularly useful for distinguishing curative from life-prolonging effects. In practice, it is common to assume a parametric model for the cure probability and a semiparametric model for the survival of the susceptibles. Because of the latent cure status, maximum likelihood estimation is performed by means of the iterative EM algorithm. Here, we focus on the cure probabilities and  propose a two-step procedure to improve upon the performance of the maximum likelihood estimator when the sample size is not large. The new method is based on the idea of presmoothing by first constructing a nonparametric estimator and then projecting it into the desired parametric class. We investigate the theoretical properties of the resulting estimator and show through an extensive simulation study for the logistic-Cox model that it outperforms the existing method. Practical use of the method is illustrated through two melanoma datasets. 
\end{abstract}



% history:
% \received{\smonth{1} \syear{0000}}

%\tableofcontents

\end{frontmatter}

	\section{Introduction}
Cure models are used to analyze time until occurrence of an event of interest when a proportion of the study population is immune to that event (cured). They are recently becoming increasingly popular in oncology as curative treatments are now a possibility, meaning that some patients will not experience cancer relapse/death (see for example \cite{legrand2019cure,othus2012cure}). More broadly, cure models find applications in studies of fertility (\cite{van2013can}), hospitalization of COVID-19 patients (\cite{pedrosa2022cure}), equipment failure in engineering (\cite{meeker1987limited}), credit scoring in economics (\cite{dirick2017time,dirick2019macro}),  etc.. What makes statistical modeling and estimation challenging when not all subjects are susceptible to the event of interest, is the unobserved cure status. As a consequence of a limited follow-up period, all cured subject are observed as censored, hence mixed with the uncured ones. 

There are two main families of cure models: promotion time models and mixture cure models (see \cite{AK2018} and \cite{peng2021cure} for an overview). The latter ones are particularly attractive in practice because, by 
 separately modeling the uncure probability (incidence) and the survival of the susceptibles (latency) given possibly different sets of covariates, they are able to distinguish a curative from a life-prolonging effect. Early works on mixture cure models were fully parametric approaches (\cite{farewell82,yamaguchi92,kuk92}), while more recently semi-parametric (\cite{ST2000,peng2000,li2002,zhang2007}) and non-parametric (\cite{XP2014,PK2019,AKL19,lopez2017incidence,lopez2017latency}) models have been proposed. Among them, the semiparametric models are often used in practice as a reasonable compromise between flexibility and simplicity. These models assume a parametric form of the incidence and a semiparametric form for the latency, with the most common choice being the mixture of the logistic with the Cox proportional hazards model (e.g. \cite{yilmaz2013insights,stringer2016cure,wycinka2017}). 

Estimation in the logistic-Cox or in general semi-parametric mixture cure models is mostly carried out via the Expectation-Maximization algorithm because of the latent cure status. Such estimators were proposed in \cite{peng2000} and \cite{ST2000} for the logistic-Cox model; in \cite{li2002,zhang2007,lu2010} for the logistic-accelerated failure time model. The procedure is implemented in the R-package \texttt{smcure} (\cite{cai_smcure}). However, for limited sample sizes which are common in practice, such iterative procedures are characterized by large mean-squared-error (MSE), convergence problems and instability of the estimators for the incidence component depending on which variables are included in the latency model (see for e.g. \cite{musta2020presmoothing}). This might lead to incorrect conclusions regarding significant effects. 

Here we propose a new second stage estimator based on presmoothing with the aim of improving upon an initially available estimator, that can for example be the \texttt{smcure} estimator, for small and moderate sample sizes. The initial estimator is used to construct a one-dimensional covariate, conditional on which we compute a nonparametric estimator of the cure probabilities. Afterwards, the nonparametric estimator is projected on the desired parametric class (for example logistic). This allows for direct estimation of the parametric incidence component despite the latent cure status.  We focus on the cure fraction, but once that is estimated, one can also fit a semiparametric model to the latency component. Compared to the method proposed in \cite{musta2020presmoothing}, this approach does not restrict us to a one-dimensional covariate and does not require multidimensional smoothing, which is essential for practical purposes. Apart from the cure model setting, the idea of constructing a parametric estimator by nonparametric estimation has been previously proposed in the context of linear regression, variable selection and functional linear regression (\cite{cristobal1987class,presmoothing_var_sel,ferraty2012presmoothing}). The novelty of our method lies in using presmoothing as a second stage estimator where a preliminary available estimator is used to reduce the covariate dimension to one. In this way we only need to choose one bandwidth independently of the number of covariates and still profit from the advantages of presmoothing: lower MSE and more stable estimators.

The paper is organized as follows. In Sections \ref{sec:model} and \ref{sec:method}, we describe the model and the estimation procedure. In Section \ref{sec:asymptotics} we show that the resulting estimator is consistent and square-root-n convergent with a Gaussian limit distribution, provided that the initial estimator is consistent. As a particular case, we focus on the logistic/Cox mixture cure model in Section~\ref{sec:simulations} and illustrate through an extensive simulation study that the proposed estimator outperforms the \texttt{smcure} estimator by significantly reducing its mean squared error. In addition, the second step using presmoothing makes the estimator more stable towards misspecifications in the latency model. Finally, in Section~\ref{sec:application}, we apply the method to two medical datasets and show that  in practice it can lead to different conclusions compared to the \texttt{smcure} estimator. 

\section{The semiparametric mixture cure model}
\label{sec:model}
Suppose we are interested in the time $T$ until a certain event happens for a mixed population of cured ($T=\infty$) and uncured ($T=T_0<\infty$) subjects. Let $B$ be  a $0$-$1$ random variable indicating the uncured status: $B=1$ for susceptible individuals and $B=0$ otherwise. Due to the limited follow-up period,  we cannot actually observe $T$ and $B$. Instead we observe a finite follow-up time $Y=\min(T,C)$ and a censoring indicator $\Delta=\1_{\{T\leq C\}}$, where $C$ denotes  the censoring time. As a result, for all the censored observations, the cure status is unknown. In the mixture cure model, the survival function of $T$ given two covariate vectors $X\in\R^p$, $Z\in\R^q$, is given by 
$$
S(t|x,z)=\p(T>t|{X=x, Z=z})=\pi_0(x)+(1-\pi_0(x))S_u(t|z),
$$
where $S_u(t|z)=\p(T>t | Z=z, B=1)$ is the survival function of the susceptibles and $\pi_0(x)=\p(B=0 | X=x)$ denotes the cure probability. 
Using two covariate vectors $X$ and $Z$  for modeling the incidence and the latency allows the cure probability and the survival of the uncured to be affected by different variables. However it does not exclude situations in which the two vectors $X$ and $Z$ are exactly the same or share some components.  

In the context of mixture cure models, the classical survival analysis assumption of independent censoring, means that $T_0\perp (C,X) | Z$ and $B\perp (C,T_0,Z)|X$, which imply that
\begin{equation}
	\label{eqn:CI1}
	T\perp C| (X,Z),
\end{equation}
(see Lemma~1 in the supplementary material of \cite{musta2020presmoothing}). As a result, we also have 
\begin{equation}
	\label{eqn:condition_X_Z}
	\p(T=\infty|X,Z)=\p(T=\infty|X) \quad \text{ and }\quad	\p(T_0\leq t |X,Z)=\p(T_0\leq t|Z).
\end{equation}

Among various modeling approaches for the incidence and the latency, the most common choice is a parametric (logistic) model for the incidence and a semiparametric (Cox or accelerated failure time) model for the latency (\cite{PK2019,burke2020likelihood,legrand2019cure,yilmaz2013insights,stringer2016cure,wycinka2017}). The popularity of such a choice is due to the simplicity and ease of interpretation. We focus on this type of model and assume that 
\[
\pi_0(x)=1-\phi(\gamma_0^Tx),
\]
where $\phi: \R\to[0,1]$ is a known function, $\gamma_0\in \R^{p+1}$ and  $\gamma^T_0$ denotes the transpose of the vector $\gamma_0$. Here the first component of $x$ is taken to be equal to 1 and the first component of $\gamma$ corresponds to the intercept. In particular, for the logistic model, we have 
\begin{equation}
	\label{eqn:logistic}
	\phi(u)=\frac{e^u}{1+e^{u}}.
\end{equation}
To check the fit of this model in practice, one can compare the prediction error with that of a more flexible single-index model as done in \cite{AKL19} or use the test proposed in \cite{muller2019goodness} which is currently developed only for a one-dimensional covariate. 

For the latency, we assume a semiparametric model $S_u(t|z) =S_u(t|z;\beta_0,\Lambda_0) $ depending on
a finite-dimensional parameter $\beta_0\in \mathcal \R^q$, and an infinite-dimensional parameter $\Lambda_0$.  
The main examples we keep in mind are the Cox proportional hazards (PH) model
\begin{equation} 
	\label{SuCox}
	S_u(t|z)
	= \exp\{-\Lambda_0(t)\exp(\beta^T_0z)\},
\end{equation}
and  the accelerated failure time model (AFT)
\[
S_u(t|z)=\exp\left\{-\Lambda_0\left(\exp\left(\beta^T_0z\right)t\right)\right\},
\]
where  $\Lambda_0$ is the baseline cumulative hazard.

The goal is to estimate the true parameters $\gamma_0$, $\beta_0$ and $\Lambda_0$ on the basis of  
$n$ i.i.d. observations $(Y_1,\Delta_1,X_1,Z_1),\dots,(Y_n,\Delta_n,X_n,Z_n)$. 
The general conditions under which the semiparametric mixture cure model is identifiable, meaning that different parameter values lead to different distributions of the observed variables $(Y,\Delta,X,Z)$, were derived in \cite{Motti} and are the following: 
\begin{itemize}
	\item[(I1)] if $\phi(\gamma,X)=\phi(\tilde\gamma,X)$ almost surely, then $\gamma=\tilde{\gamma},$
	\item[(I2)] the function $S_u(\cdot|z)$ has support $[0,\tau(z)]$,  
	\item[(I3)] $\p(C>\tau(Z)|X,Z)>0$ for almost all $X$ and $Z$,
	\item[(I4)] if, for all $t\geq 0$, we have $	S_u(t|Z;\Lambda,\beta)=S_u(t|Z;\tilde{\Lambda},\tilde{\beta}) $  almost surely, then $\Lambda=\tilde{\Lambda}$ and $\beta=\tilde{\beta}$,
\end{itemize}
In the particular case of the logistic-Cox model the conditions become:
\begin{itemize}
	\item[(I1')]  for all $x$, $ 0 <\phi(\gamma_0^Tx) < 1$,
	\item[(I2')] the function $S_u$ has support $[0; \tau_0]$ for some $\tau_0<\infty$,
	\item[(I3')] $ P(C >\tau_0|X;Z) > 0$ for almost all X and Z,
	\item[(I4')]  the matrices $Var(X)$ and $Var(Z)$ are positive definite,
\end{itemize}
(see Proposition 1 and 2 in \cite{Motti}). 
Conditions I3 and I3' are of particular importance in the context of mixture cure models and essentially tell us that, in order to correctly identify the cure proportion, we need sufficiently long follow-up beyond  the time when the events occur. In practice, this can be evaluated based on the plateau of the Kaplan-Meier estimator and the expert (medical) knowledge. 	

\section{The 2-step estimation procedure}
\label{sec:method}

Estimation in semiparametric mixture cure models is usually performed via the expectation maximization algorithm because of the latent cure status. Such method has been proposed by \cite{peng2000,ST2000} for the logistic-Cox mixture cure model, and by \cite{li2002,zhang2007} for the logistic-AFT model. The procedure is implemented in the R package \texttt{smcure} (\cite{cai_smcure}). 
Despite the simplicity of the method, simultaneous computation of  $\gamma_0$, $\beta_0$ and $\Lambda_0$  through an iterative procedure leads to several problems for finite, not large sample sizes which are commonly encountered in practice. This has been previously reported and illustrated in \cite{musta2020presmoothing,burke2020likelihood,han2017statistical}. The main concerns are the large MSE, convergence problems and instability 
of the estimator for the incidence component depending on which
variables are included in the latency model. In particular, if
the latency model is misspecified, even the estimators of the incidence parameters suffer from induced bias. To alleviate these problems, we propose the following 2-step estimation procedure that makes use of presmoothing. 

We start with some preliminary estimator $\tilde{\gamma}_n$ of $\gamma_0$. This can be any estimator that satisfies the conditions described in Section~\ref{sec:asymptotics} and in particular for the logistic-Cox or logistic-AFT model we can use the \texttt{smcure} estimator.  
We use this preliminary estimator to construct the one-dimensional index $\tilde{V}=\tilde{\gamma}'_nX$ estimating $V=\gamma'_0X$. Based on this new one-dimensional covariate, we  compute a nonparametric estimator of the cure probability for each subject defined as follows
\begin{equation}
	\label{def:hat_pi}
	\hat\pi_n(x)=\prod_{t\in\R} \left(1-\frac{\hat{H}_{1,\tilde{\gamma}_n}\left(\dd t|\tilde{\gamma}^T_nx\right)}{\hat{H}_{\tilde{\gamma}_n}\left([t,\infty)|\tilde{\gamma}^T_nx\right)}\right){,}
\end{equation}
where  {$\hat{H}_{\tilde{\gamma}_n}([t,\infty)|u)=\hat{H}_{1,\tilde{\gamma}_n}([t,\infty)|u)+\hat{H}_{0,\tilde{\gamma}_n}([t,\infty)|u)$,} $\hat{H}_{1,\tilde{\gamma}_n}(\dd t|u) = \hat{H}_{1,\tilde{\gamma}_n}((t-\dd t,t]|u)$ for small $\dd t$ and 
\[
\hat{H}_{l,\tilde{\gamma}_n}([t,\infty)|u)=\sum_{i=1}^n\frac{{k}_b\left(\tilde{\gamma}^T_nX_i-u\right)}{\sum_{j=1}^n{k}_b\left(\tilde{\gamma}^T_nX_j-u\right)}\1_{\{Y_i\geq t, \Delta_i=l\}},\quad l=0,1{,}
\] 
are estimators of
\begin{equation}
	\label{def:H}
	H_l([t,\infty)|u)=\p\left(Y\geq t,\Delta=l|{\gamma}^T_0X=u\right),
\end{equation}
and
$H([t,\infty)|u)=H_1([t,\infty)|u)+H_0([t,\infty)|u)$. Here $k$ is a one-dimensional kernel function, $b=b_n$ is a bandwidth sequence and ${k}_{{b}}(\cdot)=k(\cdot/b)/b$. 

{The estimator} $\hat\pi_n(x)$ coincides with the Beran estimator of the {conditional survival function} $S$ at the largest observed event time $Y_{(m)}$ {and does not require any specification of $\tau_0$}. {Since $\hat{H}_{1,\tilde{\gamma}_n}\left(\dd t|\tilde{\gamma}^T_nx\right)$ is different from zero only at the observed event times, computation of $\hat\pi_n(x)$ requires only a product over $t$ in the set of the observed event times.} Afterwards, we consider the logistic likelihood
\[
\hat L_{n,1}(\gamma)=\prod_{i=1}^n \phi(\gamma^TX_i)^{1-\hat\pi_n(X_i)}[1-\phi(\gamma^TX_i)]^{\hat\pi_n(X_i)}{,}
\]
and define $\hat\gamma_n$ as the maximizer of 
\begin{equation}
	\label{def:hat_L_gamma}
	\log \hat{L}_{n,1}(\gamma)=\sum_{i=1}^n\Big\{\left[1-\hat\pi_n(X_i)\right] \log  \phi(\gamma^TX_i)+\hat\pi_n(X_i)\log \left[1-\phi(\gamma^TX_i)\right]\Big\}{\tau(X_i)}.
\end{equation}
{We introduce a trimming function $\tau(\cdot)\geq 0$ to avoid regions where the density function of the index $\gamma^TX$, for $\gamma$ in a neighborhood of $\gamma_0$, approaches zero (as done for example in \cite{lopez2013single}). We discuss possible choices of $\tau(\cdot) $ in Section \ref{sec:asymptotics}.} Existence and uniqueness of $\hat\gamma_n$ hold under the same conditions as for the maximum likelihood estimator in the binary outcome regression model where $1-\hat\pi_n(X_i)$ is replaced by {the outcome $B_i$}. For example, in the logistic model,   it is required that $p<n$ and the matrix of the variables $X$ has full rank. 
Estimation of the latency component can then be performed by maximizing the likelihood of the mixture model 
\[
%L(\gamma,\beta,\Lambda)=
\prod_{i=1}^{n}\left\{\phi(\gamma^TX_i)f_u(Y_i|Z_i;\beta,\Lambda)
\right\}^{\Delta_i}\!\left\{1\!-\phi(\gamma^TX_i)+\phi(\gamma^TX_i)S_u(Y_i|Z_i;\beta,\Lambda) \right\}^{1-\Delta_i},
\]
with respect to $\beta$ and $\Lambda$ for $\gamma=\hat\gamma_n$. Here $ f_u(t|Z;\beta,\Lambda)=-\frac{\dd}{\dd t}S_u(t|Z;\beta,\Lambda)$. In practice this would mean performing the EM algorithm (as in the \texttt{smcure} package) only on the  latency component, i.e. keeping $\gamma=\hat{\gamma}_n$ fixed and updating  $\beta$ and $\Lambda$ in each iteration. 

We call this a 2-step estimator because it relies on a preliminary estimator $\tilde{\gamma}_n$, which is used to construct the one-dimensional covariate  $\tilde{\gamma}^T_nX$. In this way, independently of the dimension of $X$, the kernel estimator requires only one bandwidth parameter. The idea of a single-index structure is used in several papers to avoid multidimensional regression (e.g. \cite{strzalkowska2014beran,lopez2013single}), but has not previously  been exploited in the context of cure models and presmoothing. A nonparametric estimator for the cure probability in \eqref{def:hat_pi} could be obtained using any nonparametric estimator of the conditional survival function as done for example in \cite{pelaez2021test,pelaez2021sort} for estimation of the default probability. Here we use the Beran estimator since it is easy to compute and exhibits a good behavior.

\subsection{Rationale behind the new approach}
By definition we have 
\[
\pi_0(x)=\p(T=\infty\mid X=x)=\p\left(T=\infty\mid \gamma^T_0X=\gamma^T_0x\right).
\]
{Moreover, since from our model it follows that $T\perp C\mid\gamma^T_0X,\beta_0^TZ$}, we have
\[
\begin{aligned}
	H_1\left(\dd t|\gamma^T_0x,{\beta_0^Tz}\right)&=F_C\left([t,\infty)|\gamma^T_0x,{\beta_0^Tz}\right)F_T\left(\dd t|\gamma^T_0x,{\beta_0^Tz}\right),\\
	H\left([t,\infty)|\gamma^T_0x,{\beta_0^Tz}\right)&=F_T\left([t,\infty]|\gamma^T_0x,{\beta_0^Tz}\right)F_C\left([t,\infty)|\gamma^T_0x,{\beta_0^Tz}\right),
\end{aligned}
\]
which yields
\[
\frac{F_T\left(\dd t|\gamma^T_0x,{\beta_0^Tz} \right)}{F_T\left([t,\infty]|\gamma^T_0x,{\beta_0^Tz}\right)}=\frac{H_1\left(\dd t|\gamma^T_0x,{\beta_0^Tz} \right)}{H\left([t,\infty)|\gamma^T_0x,{\beta_0^Tz}\right)}.
\]
As in \cite{PK2019}, we obtain that 
\[
\p\left(T=\infty\mid \gamma^T_0X=\gamma^T_0x,{\beta_0^TZ=\beta_0^Tz}\right)=\prod_{t\in\R} \left\{1-\frac{H_1\left(\dd t|\gamma^T_0x,{\beta_0^Tz} \right)}{H\left([t,\infty)|\gamma^T_0x,{\beta_0^Tz}\right)}\right\},
\]
where $\prod_{t\in\R} $ denotes the product integral. By the first part of \eqref{eqn:condition_X_Z}, the product integral is also equal to $P(T = \infty | \gamma_0^T X = \gamma_0^T x)$.  {Similarly, under the slightly stronger assumption that $T\perp C\mid\gamma^T_0X$, we obtain
	\[
	\p\left(T=\infty\mid \gamma^T_0X=\gamma^T_0x\right)=\prod_{t\in\R} \left\{1-\frac{H_1\left(\dd t|\gamma^T_0x \right)}{H\left([t,\infty)|\gamma^T_0x\right)}\right\},
	\] 
}
This justifies the definition of our estimator in \eqref{def:hat_pi}. This assumption  is satisfied if {$C\perp (X,Z)\mid \gamma_0^TX$. We restrict ourselves to this case for simplicity in order to have conditioning on only one index. However, the method can be used in general conditioning on both $\gamma_0^TX$ and $\beta_0^TZ$. }
We illustrate this through one of the simulation settings in Section~\ref{sec:simulations}.

\section{Asymptotic results}
\label{sec:asymptotics}
In this section we focus on models where the survival function of the susceptibles $S_u$ has fixed support $[0,\tau_0]$ such that 
\begin{equation}\label{eqn:cond_support}
	\inf_{\gamma\in G}
	\p(C>\tau_0\mid \gamma^TX)>0 \,\text{ almost surely}.
\end{equation} This is the case when the latency follows a Cox regression model and the identifiability assumption (I3') is satisfied.  Moreover, we assume that $T\perp C\mid\gamma^T_0X$ but the results can be generalized as mentioned in the previous section.

Let us sketch the arguments we will use to obtain asymptotic properties of $\hat\gamma_n$. Note that $\hat\gamma_n$ is the maximizer of the Bernoulli-type log-likelihood in \eqref{def:hat_L_gamma}. 
%, which is the same as the one in \cite{musta2020presmoothing} (equation 10) but with a different estimator $\hat\pi_n$. Consistency and asymptotic normality of $\hat{\gamma}_n$ follow from Theorem 1 and Theorem 3 of \cite{musta2020presmoothing}, which are stated for a general estimator $\hat\pi_n$.
 Hence, the main issue is dealing with the nonparametric estimator $\hat\pi_n$, which replaces the latent binary outcome.
By definition we have
\[
\hat\pi_n(x)=\hat{g}_{\tilde{\gamma}_n}(\tilde{\gamma}^T_nx),
\]
where
\begin{equation}
	\label{def:hat_g}
	\hat{g}_\gamma(u)=1-\hat{F}_n(\tau_0\mid \gamma^TX=u)=\prod_{t\in\R} \left(1-\frac{\hat{H}_{1,{\gamma}}(\dd t|u)}{\hat{H}_{{\gamma}}([t,\infty)|u)}\right).
\end{equation}
Note that the product integral actually over $t\in\R$ is the same as over $t\leq\tau_0$ because $\hat{H}_{1,{\gamma}}(\dd t|u)=0$ for $t>Y_{(m)}$, where $Y_{(m)}\leq\tau_0$ is the last observed event time.
We can also write 
\[
\pi_0(x)=g_{\gamma_0}(\gamma^T_0x),
\]
with 
\begin{equation}
	\label{def:g}
	{g}_\gamma(u)=1-{F}_T(\tau_0\mid \gamma^TX=u)=\prod_{t\in\R} \left(1-\frac{{H}_{1}(dt|u)}{H([t,\infty)|u)}\right).
\end{equation}
Hence 
\begin{equation}
	\label{eqn:hat_pi_decomp}
	\begin{aligned}
		\hat\pi_n(x)-\pi_0(x)&=\hat{g}_{\tilde{\gamma}_n}\left(\tilde{\gamma}^T_nx\right)-g_{\gamma_0}\left(\gamma^T_0x\right)\\
		&=\left\{\hat{g}_{{\gamma}_0}\left({\gamma}^T_0x\right)-g_{\gamma_0}\left(\gamma^T_0x\right)\right\}+\left\{\hat{g}_{\tilde{\gamma}_n}\left(\tilde{\gamma}^T_nx\right)-\hat{g}_{\gamma_0}\left(\gamma^T_0x\right)\right\}\\
		&=\left\{\hat{F}_n\left(\tau_0\mid \gamma^T_0X=\gamma^T_0x\right)-{F}_T\left(\tau_0\mid \gamma^T_0X=\gamma^T_0x\right)\right\}\\
		&\qquad+\left\{\hat{g}_{\tilde{\gamma}_n}\left(\tilde{\gamma}^T_nx\right)-\hat{g}_{\gamma_0}\left(\gamma^T_0x\right)\right\}.
	\end{aligned}
\end{equation}
The first term on the right hand side of the equation, can be dealt with as usual being the difference between $F_n$ and $F$ conditionally on a one dimensional covariate $\gamma^T_0X$. The second term results from using $\tilde\gamma_n$ instead of  $\gamma_0$ when constructing the one-dimensional covariate $\tilde{\gamma}^T_nX$. The behaviour of this term depends on the properties of the preliminary estimator $\tilde{\gamma}_n$. We first formulate the results for a general prelimary estimator $\tilde{\gamma}_n$ and a general parametric function $\phi$. We then show that, for the logistic-Cox model, the maximum likelihood estimator satisfies the required conditions. 

The following assumptions are needed for consistency of $\hat{\gamma}_n$. 
\begin{itemize}
	\item[(C1)] The preliminary estimator is consistent, i.e. $\tilde{\gamma}_n-\gamma_0=o_P(1)$. 
	\item[(C2)] The parameter $\gamma_0$ lies in the interior of a compact set $G\subset\R^p$.
	\item[(C3)]There exist some constants $a>0$, $c>0$ such that
	\[
	\left|\phi(\gamma_1^Tx)-\phi(\gamma_2^Tx)\right|\leq c\Vert\gamma_1-\gamma_2\Vert^a,\qquad\forall\gamma_1,\gamma_2\in G,\,\forall x\in\X,
	\]
	where $\Vert\cdot\Vert$ denotes the Euclidean norm and $\X\subset \R^p$ is the support of $X$.
	\item[(C4)] $\inf_{\gamma\in G}\inf_{x\in\X}\phi(\gamma^Tx)>0$ and $\sup_{\gamma\in G}\sup_{x\in\X}\phi(\gamma^Tx)<1$.
	\item[(C5)] For any $\gamma\in G$, the support $\X_\gamma$ of  $\gamma'\X$ is a bounded convex subset of $\R$. The density $f_{X_\gamma}(\cdot)$ of $X_\gamma$ is twice differentiable with a bounded second derivative.
	\item[(C6)] The bandwidth $b$ is such that $nb^4\to 0$ and $nb^{3+\xi}/(\log b^{-1}) \to \infty$ for some $\xi>0$.
	\item[(C7)] The kernel $k$ is a twice continuously differentiable, symmetric  probability density function with compact support.
	\item[(C8)] (i) The functions $H([0,t]|u)$, $H_1([0,t]|u)$ defined in \eqref{def:H} are twice differentiable with respect to $u$, with uniformly bounded derivatives for all $t\leq \tau_0$, $u\in\X_{\gamma_0}$. Moreover, there exist continuous nondecreasing functions $L_1$, $L_2$, $L_3$ such that $L_i(0)=0$, $L_i(\tau_0)<\infty$ and 	for all $t, s\in[0,\tau_0]$, $u\in\X_{\gamma_0}$,
	\[
	\begin{split}
		\left|H_c(t|u)-H_c(s|u)\right|&\leq \left|L_1(t)-L_1(s)\right|,\\ \left|H_{1c}(t|u)-H_{1c}(s|u)\right|&\leq \left|L_1(t)-L_1(s)\right|,\\
		\left|\frac{\partial H_c(t|u)}{\partial u}-\frac{\partial H_c(s|u)}{\partial u}\right|&\leq \left|L_2(t)-L_2(s)\right|,\\
		\left|\frac{\partial H_{1c}(t|u)}{\partial u}-\frac{\partial H_{1c}(s|u)}{\partial u}\right|&\leq \left|L_3(t)-L_3(s)\right|,
	\end{split}
	\]
	where the subscript c denotes the continuous part of a function.
	
	(ii) The number of jump points for the distribution function $F_C(t|u)$ of the censoring times given the index $\gamma_0^TX=u$, are {finite and} the same for all $u$. The partial derivative of $F_C(t|u)$ with respect to $u$ exists and is uniformly bounded for all $t\leq\tau_0$, $u\in\X_{\gamma_0}$. Moreover, 
	the partial derivative with respect to $u$ of $F_T(t|u)$ {(distribution function  of the survival times $T$ given $\gamma_0^TX=u$)}  exists and is uniformly bounded for all $t\leq\tau_0$, $u\in\X_{\gamma_0}$.   
	\item[(C9)]The function $(x,\gamma)\mapsto g_\gamma(\gamma^Tx)$ is continuously differentiable with respect to $\gamma$ and the vector $\nabla_\gamma g_\gamma(\gamma^Tx)$ is continuous with respect to  $(x,\gamma)$.
\end{itemize}

Assumptions (C2)-(C4), (C6)-(C8) are standard assumptions (see for example \cite{PK2019,musta2020presmoothing,KA99}.  Assumptions (C5) and (C9) are needed because we compute the nonparametric estimator using the index $\tilde{\gamma}_n^TX$ instead of $\gamma_0^TX$. Such assumptions appear for example in \cite{lopez2013single}.

{A possible choice of the trimming function $\tau(\cdot)$ in \eqref{def:hat_L_gamma} could be $\tau(x)=\1_{\tilde{\X}}(x)$ if we know a set $\tilde{\X}$ such as
	\begin{equation}
		\label{eqn:trimming}
		\inf_{\gamma\in G}\inf_{x\in\tilde{\X}}f_{X_\gamma}(\gamma^Tx)=c>0.	
	\end{equation}
	Otherwise, as shown in \cite{lopez2013single}, one can take $\tau(x)=\1_{\{f_{X_{\gamma_0}}(\gamma_0^Tx)\geq c\}}$ for some $c>0$, which is asymptotically equivalent to the previous proposal. In practice, we can use $\tau(x)=\1_{\{\hat{f}_{X_{\tilde\gamma}}(\tilde\gamma_n^Tx)\geq c\}}$ based on the preliminary estimator $\tilde{\gamma}_n$. }
\begin{theo}
	\label{theo:1}
	Assume that conditions (C1)-(C9) are satisfied. Then $$\hat\gamma_n-\gamma_0=o_P(1).$$
\end{theo}

In order to obtain asymptotic normality of $\hat{\gamma}_n$ at rate $\sqrt{n}$, we need the following additional assumptions.  
\begin{itemize}
	\item[(N1)] For each $x\in\X$, the function $\gamma\mapsto\phi(\gamma^Tx)$ is twice continuously differentiable with uniformly bounded derivatives in $G\times\X$.
	%	\item[(N2)] The function $g_{\gamma_0}(\cdot): \X_{\gamma_0}\to [0,1]$ is continuously differentiable and such that $\sup_{u\in\X_{\gamma_0}}|g'_{\gamma_0}(u)|\leq M$ and 
	%	\[
	%	\sup_{u_1, u_2\in\X_{\gamma_0}}\frac{|g'_{\gamma_0}(u_1)-g'_{\gamma_0}(u_2)|}{|u_1-u_2|^\xi}\leq M
	%	\]
	%	for some $M>0$ and $\xi\in(0,1]$.
	\item[(N2)] The matrix $\E\left[\phi'(\gamma_0^TX)^2XX^T\right]$ is positive definite.
	\item[(N3)] The preliminary estimator $\tilde\gamma_n$ is $\sqrt{n}$ consistent and such that there exists a function $\zeta$ such that 
	\[
	\tilde{\gamma}_n-\gamma_0=\frac1n\sum_{i=1}^n\zeta(T_i,\Delta_i,X_i,Z_i)+R_n,
	\]
	with $\Vert R_n\Vert=o_P(n^{-1/2})$ and $\E[\zeta(T,\Delta,X,Z)]=0$.
\end{itemize}
Again (N1)-(N2) are standard assumptions, while (N3) arises from the use of the index $\tilde{\gamma}_n^TX$ instead of $\gamma_0^TX$. As a result the asymptotic variance of $\hat\gamma_n$ will also depend on the asymptotic variance of the preliminary estimator $\tilde{\gamma}_n$.
\begin{theo}
	\label{theo:2}
	Assume that conditions (C1)-(C9), (N1)-(N3) are satisfied. Then
	\[
	\sqrt{n}(\hat{\gamma}_n-\gamma_0)\xrightarrow{d}N(0,\Sigma_\gamma),
	\] 
	with covariance matrix $\Sigma_\gamma$ defined in \eqref{def:Sigma_gamma}.
\end{theo}
Given the complicated form of the covariance matrix $\Sigma_\gamma$, we suggest using a bootstrap procedure for estimating the standard errors as also done  for the maximum likelihood estimator of a semi-parametric mixture cure model.

If we consider the particular case of a logistic-Cox mixture cure model and take the maximum likelihood estimator as a preliminary estimator $\tilde{\gamma}_n$, then assumptions (C2)-(C4), (C9), (N1)-(N2) are obviously satisfied for the logistic model. Morevoer  (C1) and (N2) are satisfied if the cumulative baseline function $\Lambda_0$ is strictly increasing and continuously differentiable under the condition
\begin{equation}
	\label{eqn:jump_cond}
	\inf_{z}\p(T_0\geq\tau_0|Z=z)>0, %=\inf_{z}\p(T_0=\tau_0|Z=z)>0,
\end{equation}
(see Theorem 2 and Theorem 3 in \cite{Lu2008}). 
  Then, from Theorems \ref{theo:1} and \ref{theo:2} it follows that the 2-step estimator is also consistent and $\sqrt{n}$-convergent. If we continue estimating the latency sub-model using this estimator of $\gamma_0$, then the resulting estimator of $\beta_0$ and $\Lambda_0$ have the desired asymptotic behavior as in Theorems 2 and 4 in \cite{musta2020presmoothing}. The proof remains the same given that they only use consistency and the asymptotic i.i.d. expression of the estimator as in assumption (N3). 
\section{Simulation study}
\label{sec:simulations}
In this section we investigate the finite-sample behaviour of the 2-step approach in the logistic/Cox mixture cure model and compare it with the maximum likelihood estimator implemented in the R package \texttt{smcure}. We use the \texttt{smcure} estimator as preliminary estimator $\tilde{\gamma}_n$ for the new method. 

We make some standard and common choices when computing the nonparametric estimator in~\eqref{def:hat_pi}. 
The kernel function $k$ is taken to be the Epanechnikov kernel $k(u)=(3/4)(1-u^2)\1_{\{|u|\leq1\}}
$. Using the preliminary estimator $\tilde{\gamma}_n$, we compute the smoothing bandwidth by cross-validation as implemented in the R package \texttt{np} for kernel estimators of conditional distribution functions, in our case for the  estimation of $H=H_0+H_1$ given $\tilde{\gamma}^TX$. In addition, we restrict ourselves to the interval  $[0,Y_{(m)}]$, where $Y_{(m)}$ is the last observed event time since the estimator of the cure probability $\hat\pi$ in \eqref{def:hat_pi} is essentially a product over values of $t$ that are equal to the observed event times. This means that we use the {cross-validation} bandwidth for estimating the conditional distribution $H(t|\tilde{\gamma}^Tx)$ for $t\leq Y_{(m)}$. 

{We could use a trimming function  $$\tau(x)=\1_{\{\hat{f}_{X_{\tilde\gamma}}(\tilde\gamma_n^Tx)\geq c\}},$$ for some small value of $c$ as proposed in Section \ref{sec:asymptotics}. However, we observe that in practice this does not affect the results since $c$ can be chosen as arbitrarily small. Hence, we do not do any trimming so that we do not have to worry about the choice of the trimming constant. The trimming is mainly introduced for the asymptotic study in order to avoid the assumption that the density of the index is bounded from below by a positive constant.}

We consider four different models and for each of them, three scenarios, covering a wide range of settings with different number and choice of covariates (continuous and discrete), different  cure and censoring rate, and different censoring mechanisms (independent of covariates, depending on the same index as the incidence model, depending on both indexes of the incidence and latency). The models are as follows. 

\textit{Model 1.} Both incidence and latency depend on two  independent continuous covariates $X_1=Z_1\sim N(0,1)$ and $X_2=Z_2\sim \text{Unif}(-1,1)$. We generate the cure status $B$ as a Bernoulli random variable with success probability  $\phi(\gamma^TX)$ where $\phi$ is the logistic function in \eqref{eqn:logistic} and $\gamma=(\gamma_0,1.5,1.5)$. The survival times for the uncured observations are generated according to a Weibull proportional hazards model
\[
S_u(t|z)=\exp\left(-\mu t^\rho\exp\left(\beta^Tz\right)\right),
\]
and are truncated at $\tau_0=15$ for $\rho=0.75$, $\mu=1.5$ and $\beta=(0.5,0.3)$. The censoring times are independent from $X$ and $T$. They are generated from the exponential distribution with parameter $\lambda_C$ and are truncated at $\tau=17$. 

\textit{Model 2.} Both incidence and latency depend on three independent covariates $X_1=Z_1\sim N(0,1)$, $X_2=Z_2\sim \text{Bernoulli}(0.3)$ and $X_3=Z_3\sim \text{Bernoulli}(0.7)$. The cure status  and the survival times for the uncured observations are generated as in Model 1 for $\gamma=(\gamma_0,-1,1,-0.3)$,  $\beta=(-0.8,1.5,-0.5)$, $\rho=0.75$, $\mu=1.5$, and $\tau_0=7$. The censoring times are  generated according to a Weibull proportional hazards model
\[
S_C(t|x)=\exp\left(-\lambda_C\mu t^{\rho}\exp\left(\gamma^Tx\right)\right),
\] 
for various choices of $\lambda_C$ and are truncated at $\tau=9$. 

\textit{Model 3.} For the incidence we consider four independent covariates: $X_1\sim N(0,1)$, $X_2\sim\text{Unif}(-1,1)$, $X_3$ and $X_4$ are Bernoulli random variables with parameters $0.4 $ and  $0.6$ respectively. The latency  depends on three covariates: $Z_1\sim N(0,1)$, $Z_2=X_2$  and $Z_3=X_4$.  The cure status  and the survival times for the uncured observations are generated as in Model 1 for $\gamma=(\gamma_0,-0.3,0.8,0.5,-1)$, $\rho=0.75$, $\mu=1.5$, $\beta=(0.1,0.4,-0.2)$ and $\tau_0= 10$.  The censoring times are  generated  generated according to a Weibull proportional hazards model
\[
S_C(t|x)=\exp\left(-\lambda_C\mu t^{\rho}\exp\left(0.4\gamma^Tx+0.5\beta^Tz\right)\right),
\] 
for various choices of $\lambda_C$ and are truncated at $\tau=12$.  

\textit{Model 4.} For the incidence we consider five independent covariates: $X_1\sim N(0,1)$, $X_2\sim\text{Unif}(-1,1)$, $X_3$ is Binomial with parameters $2$ and $0.5$, $X_4$ and $X_5$ are Bernoulli random variables with parameters $0.4 $ and  $0.6$ respectively. The latency  depends on three covariates: $Z_1\sim N(0,1)$, $Z_2=X_3$  and $Z_3=X_4$.  The cure status  and the survival times for the uncured observations are generated as in Model 1 for $\gamma=(\gamma_0,-0.8,0.3,-0.4,0.5,0.6)$, $\rho=0.75$, $\mu=1.5$, $\beta=(0.2,-0.5,0.3)$ and $\tau_0= 7$. 
The censoring times are independent from $X$, $Z$ and $T$. They are generated from the exponential distribution with parameter $\lambda_C$ and are truncated at $\tau=9$. 

For the four models we choose the values of the unspecified parameters $\gamma_0$ and $\lambda_C$ in such a way that the cure rate is around $20\%$, $40\%$ or $60\%$ and the difference between the cure and the censoring rate is around $5\%$, $10\%$ or $15\%$. The specification of the parameters and the corresponding censoring and cure rates are given Table~\ref{tab:models}. The truncation of the  survival  and censoring times on $[0,\tau_0]$ and $[0,\tau]$ is made in such a way that $\tau_0<\tau$ and  condition \eqref{eqn:jump_cond} is satisfied but in practice
it is unlikely to observe event times at $\tau_0$. In this way, we try to find a compromise between theoretical assumptions and  real-life scenarios. Model 3 illustrates the behavior of the method when the censoring times depends on both the indexes of the incidence and latency models (which was assumed for simplicity in the theoretical study). 

\begin{table}
	\caption{	\label{tab:models}Parameter values and model characteristics for each scenario.}
	\centering
	\addtolength{\tabcolsep}{-4pt}
	\fbox{%
		\begin{tabular}{ccccccc}
			Model & Scenario & $\gamma_0$ & &  $\lambda_C$ & Cens. rate & Cure rate\\
			\hline
			& & & & & & \\[-10pt]
			& $1 $ &$2$ & & $0.4 $ &$36\%$ & $20\% $\\
			%	\cline{2-6}
			1& $2$ & $0.6$  &  & $0.4 $ & $50\%$ & $40\% $\\
			%	\cline{2-6}
			& $3 $	& $-0.5$  & & $0.3$ & $63\%$ & $58\% $\\
			\hline
			& & & & & &  \\[-10pt]
			& $1$	& $1.6$ 
			&   & $1/35 $ & $30\%$ & $20\% $\\
			%			\cline{2-6}
			2& $2 $  & $0.4 $ &  & $1/20 $ & $45\%$ & $40\% $\\
			%	\cline{2-6}
			& $3 $	& $-0.6 $  & & $1$ & $75\%$ & $60\% $\\
			\hline
			& & & & &  & \\[-10pt]
			& $1$	& $2$ &   &  $1/9$ & $35\%$ & $20\% $\\
			%	\cline{2-6}
			3	& $2$& $0.9$& &    $1/7 $ & $50\%$ & $40\% $\\
			%	\cline{2-6}
			&  $3$	&$-0.1 $ & & $1/7 $ & $65\%$ & $60\% $\\
			\hline
			& & & & &  &\\[-10pt]
			& $1$ 	& $1.5 $    && $0.6$ & $25\%$ & $20\% $\\
			%	\cline{2-6}
			4	& $2$	& $0.3 $  & & $0.3$ & $55\%$ & $40\% $\\
			%	\cline{2-6}
			& $3$ 	&$-0.6 $ & & $0.4 $ & $70\%$ & $60\% $\\
			%\hline
		\end{tabular}
	}
\end{table}
We consider samples of size $n=200$ and $n=400$ since we aim to provide a method that improves upon the maximum likelihood estimator for small and moderate sample size.
For each configuration, $1000$  datasets were generated and the estimators of $\beta$ and $\gamma$ were computed through \texttt{smcure} and the proposed 2-step approach. We report the bias, variance and mean squared error (MSE) of the estimators, computed over the iterations for which the \texttt{smcure} procedure converges, in Tables~\ref{tab:results1}-\ref{tab:results3}. In some scenarios, mainly corresponding to the ones with $15\%$ additional censoring compared to the cure rate and smaller sample size, the iterative procedure of the EM algorithm in \texttt{smcure} does not converge. The most problematic setting in this regard is Model 2 scenario 3, for which $73/1000$ iterations do not converge for $n=200$ and $36/1000$ for $n=400$. The boxplots of the estimators for both methods in these non-convergent iterations are shown in Figure~\ref{fig:non_conv}. In the other settings, only around $2\%$ or less of the iterations do not converge.  

\begin{table}
	\caption{\label{tab:results1}Bias, variance and MSE of $\hat\gamma$ and $\hat\beta$ for \texttt{smcure} (second rows) and the 2-step approach (first rows) in Models 1 and 2.}
	\centering
	%	\addtolength{\tabcolsep}{-2pt}
	{\small		\scalebox{0.85}{
			\fbox{
				\begin{tabular}{cccrrrrrrrrr}
					%	\hline
					%			& & & & & & & & & && \\[-10pt]
					%	&&&\multicolumn{9}{c}{Censoring scheme}\\
					%	\cline{4-12}
					&	& & & & & & & & && \\[-7pt]
					&	&&\multicolumn{3}{c}{Scenario 1}&\multicolumn{3}{c}{Scenario 2}&\multicolumn{3}{c}{Scenario 3}\\
					Mod.&	n & Par. &  Bias & Var. & MSE & Bias & Var. & MSE & Bias & Var. & MSE\\[2pt]
					\hline
					&	 & & & & & & & & && \\[-8pt]
					1&	$200$ & $\gamma_1 $ & $-0.007 $  & $0.246 $ & $0.246  $ & $-0.013  $ & $0.062 $ & $0.062 $ & $0.088  $ & $0.049 $ & $0.057 $\\
					& &  & $0.243 $  & $0.534 $ & $ 0.593 $ & $0.058 $ & $0.092$ & $0.095$ & $  0.096$ & $0.053 $ & $0.062$\\
					& & $\gamma_2 $ & $-0.009 $  & $0.217$ & $0.217 $ & $ -0.035$ & $0.116$ & $0.117 $ & $-0.043 $ & $0.089$ & $0.091 $\\
					& & & $ 0.201$  & $0.302 $ & $0.343  $ & $ 0.096 $ & $0.146 $ & $0.155 $ & $0.049 $ & $0.099 $ & $0.102$\\
					& & $\gamma_3 $ & $-0.050 $  & $0.406$ & $0.408  $ & $ -0.028 $ & $0.205$ & $0.206 $ & $-0.031  $ & $0.152 $ & $0.153 $\\
					& & & $ 0.189$  & $0.717 $ & $0.753  $ & $ 0.108 $ & $0.242 $ & $0.253 $ & $0.063  $ & $0.174$ & $0.178 $\\
					& & $\beta_1 $ & $-0.006$  & $0.015 $ & $0.015  $ & $0.005  $ & $ 0.024$ & $ 0.024$ & $0.015 $ & $0.034$ & $0.034 $\\
					& &  & $-0.009 $  & $0.016$ & $ 0.016 $ & $- 0.002 $ & $0.025 $ & $0.025 $ & $0.009  $ & $0.034 $ & $0.034 $\\
					& & $\beta_2 $ & $0.008$  & $0.035$ & $0.035 $ & $0.008  $ & $ 0.050$ & $ 0.050$ & $0.012  $ & $0.071$ & $0.072 $\\
					& &  & $0.003 $  & $0.036$ & $ 0.036 $ & $ -0.002 $ & $0.052 $ & $0.052 $ & $0.002  $ & $0.073 $ & $0.073$\\
					\cline{2-12}
					&&&&&&&&&&&\\[-8pt]
					&	$400$ & $\gamma_1 $ & $-0.004$  & $0.109 $ & $0.109  $ & $-0.002  $ & $0.036$ & $0.036 $& $0.000  $ & $0.026$ & $0.026 $ \\
					& &  & $0.129 $  & $0.161 $ & $ 0.178 $ & $  0.040$ & $0.043 $ & $0.045$& $0.004  $ & $0.027 $ & $0.027$ \\
					& & $\gamma_2 $ & $-0.015 $  & $0.094 $ & $0.095 $  & $-0.023  $ & $0.059 $ & $0.059 $& $ -0.012$ & $0.047 $ & $0.048$\\
					& & & $ 0.101$  & $0.108$ & $0.118  $  & $0.059  $ & $0.059$ & $0.063 $& $ 0.057 $ & $0.052 $ & $0.055 $\\
					& & $\gamma_3 $ & $-0.008$  & $0.201 $ & $0.201 $  & $-0.019 $ & $0.110$ & $0.110$& $ -0.010$ & $0.084 $ & $0.084 $\\
					& & & $ 0.116$  & $0.247 $ & $0.261  $ & $0.065 $ & $0.115 $ & $0.119 $ & $ 0.060 $ & $0.088 $ & $0.091 $\\
					& & $\beta_1 $ & $0.005$  & $0.008 $ & $0.008  $ & $0.013  $ & $0.012$ & $0.012 $ & $0.011  $ & $ 0.017$ & $ 0.017$\\
					& &  & $0.003 $  & $0.008 $ & $ 0.008 $ & $0.009  $ & $0.012 $ & $0.012 $ & $ 0.006$ & $0.017 $ & $0.017 $\\
					& & $\beta_2 $ & $0.007 $  & $0.019 $ & $0.019  $  & $0.012 $ & $0.026 $ & $0.027 $& $0.015  $ & $ 0.034$ & $ 0.034$\\
					& &  & $0.004 $  & $0.019 $ & $0.019 $  & $0.005 $ & $0.027$ & $0.027 $& $ 0.008$ & $0.035 $ & $0.035 $\\
					\cline{1-12}
					& & & & & & & & & && \\[-8pt]
					$2 $	& $200$ & $\gamma_1 $ & $0.049 $  & $0.248 $ & $0.250  $ & $0.011  $ & $ 0.114$ & $ 0.114$ & $ 0.000 $ & $0.244 $ & $ 0.244$\\
					& &  & $0.103 $  & $0.268 $ & $0.278 $ & $0.021  $ & $0.121$ & $0.122 $ & $0.097$ & $0.389$ & $0.398 $\\
					& & $\gamma_2 $ & $0.017$  & $0.073 $ & $0.074 $ & $0.010 $ & $ 0.047$ & $0.047 $ & $ 0.090 $ & $ 0.091$ & $0.099 $\\
					& & & $-0.081 $  & $ 0.088$ & $ 0.094 $ & $-0.040 $ & $0.051$ & $0.052 $ & $-0.154 $ & $0.169 $ & $0.193 $\\
					& & $\gamma_3 $ & $0.025 $  & $0.947$ & $0.948 $ & $-0.009  $ & $ 0.171$ & $0.171$ & $ -0.102 $ & $ 0.299$ & $0.310$\\
					& & & $0.167 $  & $ 1.895$ & $ 1.923 $ & $0.035 $ &$0.177$ & $0.178 $ & $0.084 $ & $0.444 $ & $0.451 $\\
					& & $\gamma_4 $ & $0.007$  & $0.273 $ & $0.273  $ & $0.005 $ & $ 0.150$ & $0.150$ & $ 0.031 $ & $ 0.270$ & $0.271$\\
					& & & $-0.021 $  & $ 0.303$ & $ 0.303$ & $-0.011  $ & $0.162 $ & $0.162 $ & $-0.040 $ & $0.459 $ & $0.461 $\\
					& & $\beta_1 $ & $ -0.003$  & $0.014$ & $ 0.014 $ & $ -0.010 $ & $ 0.021$ & $0.021$ & $-0.074 $ & $0.069 $ & $0.075 $\\
					& &  & $ 0.002$  & $0.014 $ & $ 0.014 $ & $-0.007 $ & $0.021 $ & $0.021 $ & $-0.016  $ & $0.085 $ & $0.086$\\
					& & $\beta_2 $ & $ 0.029$  & $0.050 $ & $ 0.051 $ & $ 0.043$ & $ 0.069$ & $0.071 $ & $0.099 $ & $0.236$ & $0.245$\\
					& &  & $ 0.027$  & $0.051 $ & $ 0.051 $ & $0.042  $ & $0.069 $ & $0.071 $ & $0.074  $ & $0.269 $ & $0.274 $\\
					& & $\beta_3 $ & $ -0.009$  & $0.044 $ & $ 0.044 $ & $ -0.014 $ & $ 0.052$ & $0.052 $ & $-0.034  $ & $0.220$ & $0.221 $\\
					& &  & $ -0.008$  & $0.045 $ & $ 0.045$ & $-0.013  $ & $0.052 $ & $0.052 $ & $-0.028  $ & $0.252 $ & $0.253 $\\
					\cline{2-12}
					& & & & & & & & & && \\[-8pt]
					& $400$& $\gamma_1 $ & $0.029 $  & $0.111 $ & $0.112 $ & $0.010 $ & $ 0.055$ & $ 0.055$ & $ 0.017 $ & $0.107 $ & $ 0.107$\\
					& &  & $0.055 $  & $0.111 $ & $0.114  $ & $0.013 $ & $0.056 $ & $0.056 $ & $0.058  $ & $0.129 $ & $0.133$\\
					& & $\gamma_2 $ & $0.016$  & $0.040 $ & $0.040  $ & $0.006  $ & $ 0.022$ & $0.022 $ & $ 0.061 $ & $ 0.046$ & $0.049 $\\
					& & & $-0.043$  & $ 0.039$ & $ 0.041 $ & $-0.025  $ & $0.022 $ & $0.023 $ & $-0.082  $ & $0.064 $ & $0.071 $\\
					& & $\gamma_3 $ & $-0.030 $  & $0.155 $ & $0.156  $ & $-0.017 $ & $ 0.079$ & $0.080 $ & $ -0.111 $ & $ 0.142$ & $0.154 $\\
					& & & $0.030 $  & $ 0.160$ & $ 0.161 $ & $0.012  $ & $0.084 $ & $0.084 $ & $0.006 $ & $0.163$ & $0.163 $\\
					& & $\gamma_4 $ & $0.010 $  & $0.114 $ & $0.114  $ & $0.015  $ & $ 0.066$ & $0.067 $ & $ 0.030 $ & $ 0.128$ & $0.129$\\
					& & & $-0.006 $  & $ 0.119$ & $ 0.119 $ & $0.007  $ & $0.068 $ & $0.068 $ & $-0.003 $ & $0.167$ & $0.167 $\\
					& & $\beta_1 $ & $ 0.004$  & $0.007 $ & $ 0.007 $ & $ -0.001 $ & $ 0.008$ & $0.008 $ & $-0.037 $ & $0.030 $ & $0.032 $\\
					& &  & $ 0.007$  & $0.007 $ & $ 0.007 $ & $0.001 $ & $0.008$ & $0.008 $ & $0.004  $ & $0.036 $ & $0.036$\\
					& & $\beta_2 $ & $ 0.010$  & $0.024$ & $ 0.024 $ & $ 0.018 $ & $ 0.029$ & $0.029 $ & $0.063 $ & $0.104 $ & $0.108$\\
					& &  & $ 0.009$  & $0.024 $ & $ 0.024 $ & $0.017  $ & $0.029 $ & $0.029 $ & $0.043  $ & $0.112 $ & $0.114 $\\
					& & $\beta_3 $ & $ 0.001$  & $0.021 $ & $ 0.021$ & $ -0.005 $ & $ 0.026$ & $0.026$ & $-0.015  $ & $0.100 $ & $0.100 $\\
					& &  & $ 0.001$  & $0.022 $ & $ 0.022 $ & $-0.005 $ & $0.026 $ & $0.026 $ & $-0.012  $ & $0.112 $ & $0.112$\\
				\end{tabular}
			}
	}}
\end{table}

\begin{table}
	\caption{\label{tab:results2}Bias, variance and MSE of $\hat\gamma$ and $\hat\beta$ for \texttt{smcure} (second rows) and the 2-step approach (first rows) in Models 3-4 for $n=200$.}
	\centering
	%	\addtolength{\tabcolsep}{-2pt}
	\scalebox{0.85}{
		\fbox{
			\begin{tabular}{cccrrrrrrrrr}
				%	\hline
				%			& & & & & & & & & && \\[-10pt]
				%	&&&\multicolumn{9}{c}{Censoring scheme}\\
				%	\cline{4-12}
				& & & & & & & & & && \\[-7pt]
				&&&\multicolumn{3}{c}{Scenario 1}&\multicolumn{3}{c}{Scenario 2}&\multicolumn{3}{c}{Scenario 3}\\
				Mod.&	n & Par. &  Bias & Var. & MSE & Bias & Var. & MSE & Bias & Var. & MSE\\[2pt]
				\hline
				& & & & & & & & & && \\[-8pt]
				3	&	$200$  & $\gamma_1 $ & $-0.056 $  & $0.287 $ & $0.290  $ & $ -0.016$ & $0.123 $ & $0.123 $ & $-0.010  $ & $0.097 $ & $0.097$ \\
				&	&  & $0.431$  & $4.843$ & $5.028  $ & $0.068 $ & $0.144 $ & $ 0.149$ & $0.011  $ & $0.098 $ & $0.099 $\\
				& & $\gamma_2 $ & $0.083 $  & $0.041$ & $ 0.048 $ & $0.043 $ & $ 0.033$ & $ 0.035$ & $ 0.022 $ & $0.030 $ & $0.031 $\\
				& & & $ -0.029$  & $0.071$ & $0.072  $ & $ -0.014 $ & $0.040$ & $0.040 $ & $-0.012  $ & $0.032 $ & $0.032 $\\
				& & $\gamma_3 $ & $-0.049$  & $0.177 $ & $ 0.179$ & $-0.021 $ & $ 0.105$ & $ 0.106$ & $ 0.004$ & $0.093 $ & $0.093 $\\
				& & & $ 0.105$  & $0.268$ & $0.279 $ & $ 0.049 $ & $0.120 $ & $0.123 $ & $0.043 $ & $0.097 $ & $0.099$\\
				& & $\gamma_4 $ & $-0.080 $  & $0.216 $ & $ 0.222 $ & $-0.041  $ & $ 0.131$ & $ 0.133$ & $ -0.047$ & $0.122$ & $0.124 $\\
				& & & $ 0.159$  & $1.233 $ & $1.259  $ & $ 0.046$ & $0.151 $ & $0.154 $ & $0.010  $ & $0.131 $ & $0.131$\\
				& & $\gamma_5 $ & $0.123$  & $0.325 $ & $ 0.340 $ & $0.051 $ & $ 0.158$ & $ 0.161$ & $ 0.021 $ & $0.136 $ & $0.137 $\\
				& & & $- 0.359$  & $4.817 $ & $4.946  $ & $ -0.062 $ & $0.183 $ & $0.187 $ & $-0.036 $ & $0.140 $ & $0.141 $\\
				& & $\beta_1 $ & $0.003 $  & $0.011 $  & $0.011 $ & $0.003 $ & $0.016 $ & $0.016  $ & $0.016 $ & $ 0.025$ & $0.025$\\
				& &  & $0.003 $  & $0.012 $ & $0.012  $ & $0.003  $ & $0.017 $ & $0.017 $ & $ 0.016 $ & $ 0.025$ & $0.025 $\\
				& & $\beta_2 $ & $0.022$  & $0.034
				$  & $0.034  $ & $0.017 $ & $0.045 $ & $0.046  $ & $0.028 $ & $ 0.075$ & $0.076$\\
				& &  & $0.010 $  & $0.036$ & $0.036  $ & $0.008  $ & $0.048 $ & $0.048 $ & $ 0.023$ & $ 0.077$ & $0.077 $\\
				& & $\beta_3 $ & $-0.048 $  & $0.044 $  & $0.046  $ & $-0.042 $ & $0.060 $ & $0.061 $ & $-0.030$ & $ 0.093$ & $0.094$\\
				& &  & $-0.032 $  & $0.048 $ & $0.049 $ & $-0.027  $ & $0.063 $ & $0.064 $ & $ -0.023 $ & $ 0.095$ & $0.095$\\
				\cline{1-12}
				& & & & & & & & & && \\[-8pt]
				$4$&	$200$  & $\gamma_1 $ & $0.086 $  & $0.237 $ & $0.245  $ & $ 0.037 $ & $0.238 $ & $0.239 $ & $0.010  $ & $0.229 $ & $0.229$ \\
				&	 &  & $0.110 $  & $0.273 $ & $0.285  $ & $0.077  $ & $0.311$ & $ 0.317$ & $0.013  $ & $0.310 $ & $0.310$\\
				&	 & $\gamma_2 $ & $0.010 $  & $0.058 $ & $ 0.058$ & $0.042  $ & $ 0.065$ & $ 0.067$ & $ 0.006$ & $0. 065$ & $0.069$\\
				&	 & & $- 0.069$  & $0.064 $ & $0.069  $ & $ -0.091$ & $0.077 $ & $0.085 $ & $-0.073  $ & $0.068 $ & $0.073 $\\
				&	 & $\gamma_3 $ & $-0.025 $  & $0.131 $ & $ 0.132 $ & $-0.018  $ & $ 0.122$ & $ 0.122$ & $ -0.021 $ & $0.122 $ & $0.122 $\\
				&	 & & $ 0.000$  & $0.153 $ & $0.153  $ & $ 0.033 $ & $0.162 $ & $0.164 $ & $0.027  $ & $0.155 $ & $0.156 $\\
				&	 & $\gamma_4 $ & $0.008$  & $0.095 $ & $ 0.095 $ & $0.007 $ & $ 0.108$ & $ 0.109$ & $ -0.007 $ & $0.110 $ & $0.110 $\\
				&	 & & $- 0.025$  & $0.110 $ & $0.111 $ & $ -0.049 $ & $0.138 $ & $0.140 $ & $-0.075  $ & $0.156 $ & $0.161$\\
				&	 & $\gamma_5 $ & $-0.026 $  & $0.212 $ & $ 0.212 $ & $-0.055  $ & $ 0.199$ & $ 0.202$ & $ -0.038 $ & $0.193 $ & $0.194 $\\
				&	 & & $ 0.019$  & $0.245 $ & $0.245 $ & $ 0.020$ & $0.287 $ & $0.287 $ & $0.045 $ & $0.268 $ & $0.270 $\\
				&	 & $\gamma_6 $ & $-0.024 $  & $0.176 $ & $ 0.177$ & $-0.026 $ & $ 0.172$ & $ 0.172$ & $ -0.038 $ & $0.163$ & $0.165 $\\
				&	 & & $ 0.031$  & $0.203 $ & $0.204  $ & $ 0.068$ & $0.213 $ & $0.217 $ & $0.063  $ & $0.209 $ & $0.213 $\\
				&	 & $\beta_1 $ & $0.004 $  & $0.008 $  & $0.008 $ & $0.000 $ & $0.019 $ & $0.019  $ & $-0.001 $ & $ 0.035$ & $0.035$\\
				&	 &  & $0.004 $  & $0.008 $ & $0.008 $ & $0.003  $ & $0.019 $ & $0.019$ & $ 0.001 $ & $ 0.037$ & $0.037$\\
				&	 & $\beta_2 $ & $-0.011$  & $0.019$  & $0.019  $ & $-0.001$ & $0.049$ & $0.049  $ & $0.010 $ & $ 0.095$ & $0.096$\\
				&	 &  & $-0.011 $  & $0.019$ & $0.019  $ & $0.004  $ & $0.053 $ & $0.053 $ & $ 0.021 $ & $ 0.103$ & $0.104 $\\
				&	 & $\beta_3 $ & $0.000 $  & $0.033 $  & $0.033  $ & $0.014 $ & $0.083$ & $0.083  $ & $0.021 $ & $ 0.144$ & $0.144$\\
				&	 &  & $-0.001 $  & $0.033 $ & $0.033  $ & $0.006  $ & $0.088$ & $0.088$ & $ 0.007 $ & $ 0.155$ & $0.155$\\
			\end{tabular}
	}}
\end{table}
\begin{figure}
	\centering
	\makebox{
		\includegraphics[width=0.48\linewidth]{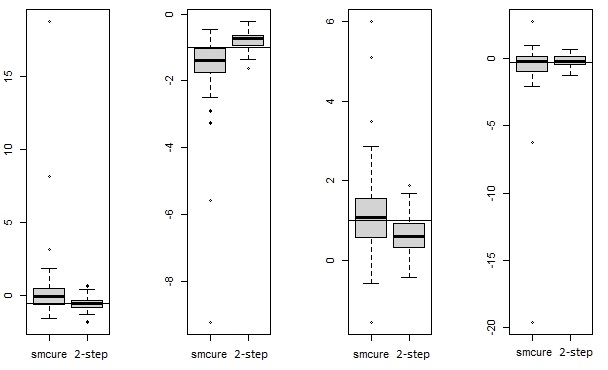}\qquad\includegraphics[width=0.48\linewidth]{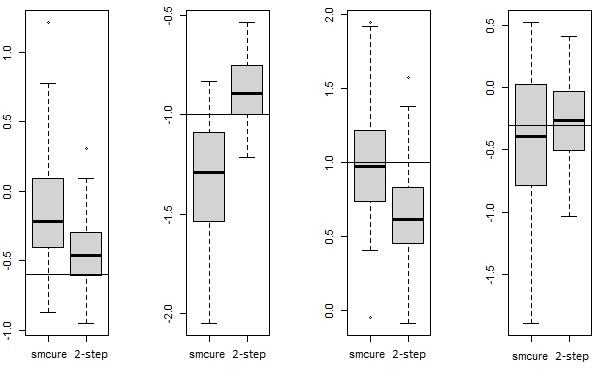}}
	\caption{\label{fig:non_conv}Boxplots of estimates of the four $\gamma$ parameters in Model 2, scenario 3 for the iterations in which \texttt{smcure} does not converge. The horizontal lines correspond to the true values of the parameters. Left panel: for $n=200$. Right panel: for $n=400$.}
\end{figure}
\begin{table}
	\caption{\label{tab:results3}Bias, variance and MSE of $\hat\gamma$ and $\hat\beta$ for \texttt{smcure} (second rows) and the 2-step approach (first rows) in Models 3-4 for $n=400$.}
	\centering
	%	\addtolength{\tabcolsep}{-2pt}
	\scalebox{0.85}{
		\fbox{
			\begin{tabular}{cccrrrrrrrrr}
				%	\hline
				%			& & & & & & & & & && \\[-10pt]
				%	&&&\multicolumn{9}{c}{Censoring scheme}\\
				%	\cline{4-12}
				& & & & & & & & & && \\[-7pt]
				&&&\multicolumn{3}{c}{Scenario 1}&\multicolumn{3}{c}{Scenario 2}&\multicolumn{3}{c}{Scenario 3}\\
				Mod.&	n & Par. &  Bias & Var. & MSE & Bias & Var. & MSE & Bias & Var. & MSE\\[2pt]
				\hline
				& & & & & & & & & && \\[-8pt]
				3	&	$400$  & $\gamma_1 $ & $-0.045 $  & $0.149 $ & $0.151  $ & $ -0.009$ & $0.069 $ & $0.070$ & $-0.020  $ & $0.046 $ & $0.047$ \\
				&	&  & $0.124$  & $0.213$ & $0.228  $ & $0.043 $ & $0.071 $ & $ 0.073$ & $-0.003  $ & $0.045 $ & $0.045 $\\
				& & $\gamma_2 $ & $0.063 $  & $0.021$ & $ 0.025 $ & $0.034 $ & $ 0.016$ & $ 0.017$ & $ 0.023 $ & $0.015 $ & $0.015 $\\
				& & & $ -0.018$  & $0.030$ & $0.030  $ & $ -0.008 $ & $0.018$ & $0.018 $ & $-0.003  $ & $0.016 $ & $0.016 $\\
				& & $\gamma_3 $ & $-0.041$  & $0.098$ & $ 0.099$ & $-0.025 $ & $ 0.058$ & $ 0.059$ & $ 0.001$ & $0.050 $ & $0.050 $\\
				& & & $ 0.052$  & $0.116$ & $0.119 $ & $ 0.019 $ & $0.061 $ & $0.062 $ & $0.026 $ & $0.052 $ & $0.052$\\
				& & $\gamma_4 $ & $-0.088 $  & $0.111 $ & $ 0.119 $ & $-0.052  $ & $ 0.069$ & $ 0.071$ & $ -0.020$ & $0.059$ & $0.059 $\\
				& & & $ 0.029$  & $0.149 $ & $0.149  $ & $ 0.007$ & $0.076$ & $0.076$ & $0.019  $ & $0.060 $ & $0.060$\\
				& & $\gamma_5 $ & $0.103$  & $0.165 $ & $ 0.175 $ & $0.047 $ & $ 0.083$ & $ 0.085$ & $ 0.033 $ & $0.068 $ & $0.069 $\\
				& & & $- 0.079$  & $0.226 $ & $0.232 $ & $ -0.030$ & $0.085 $ & $0.086 $ & $-0.013$ & $0.066 $ & $0.067 $\\
				& & $\beta_1 $ & $-0.001 $  & $0.005$  & $0.005 $ & $0.001 $ & $0.007 $ & $0.007  $ & $0.003$ & $ 0.010$ & $0.010$\\
				& &  & $-0.001 $  & $0.005$ & $0.005  $ & $0.001  $ & $0.007 $ & $0.007 $ & $ 0.003 $ & $ 0.010$ & $0.010 $\\
				& & $\beta_2 $ & $0.007$  & $0.017
				$  & $0.017  $ & $0.006 $ & $0.023 $ & $0.023  $ & $0.002 $ & $ 0.034$ & $0.034$\\
				& &  & $0.000 $  & $0.018$ & $0.018  $ & $0.001 $ & $0.024 $ & $0.024 $ & $ -0.001$ & $ 0.035$ & $0.035 $\\
				& & $\beta_3 $ & $-0.017 $  & $0.021 $  & $0.021  $ & $-0.012 $ & $0.028 $ & $0.028 $ & $-0.007$ & $ 0.042$ & $0.042$\\
				& &  & $-0.002 $  & $0.022 $ & $0.022 $ & $-0.003 $ & $0.029 $ & $0.029 $ & $ -0.002 $ & $ 0.043$ & $0.043$\\
				\cline{1-12}
				& & & & & & & & & && \\[-8pt]
				$4$&	$400$  & $\gamma_1 $ & $0.042 $  & $0.115$ & $0.116  $ & $ 0.016 $ & $0.105 $ & $0.105 $ & $0.003  $ & $0.118 $ & $0.118$ \\
				&	 &  & $0.056 $  & $0.124 $ & $0.127  $ & $0.036  $ & $0.119$ & $ 0.121$ & $0.007  $ & $0.125 $ & $0.125$\\
				&	 & $\gamma_2 $ & $0.025 $  & $0.027$ & $ 0.027$ & $0.045  $ & $ 0.036$ & $ 0.038$ & $ 0.047$ & $0. 036$ & $0.039$\\
				&	 & & $- 0.029$  & $0.028 $ & $0.029  $ & $ -0.040$ & $0.032 $ & $0.034 $ & $-0.038  $ & $0.032 $ & $0.034 $\\
				&	 & $\gamma_3 $ & $0.002 $  & $0.064$ & $ 0.064 $ & $-0.022 $ & $ 0.062$ & $ 0.062$ & $ -0.018 $ & $0.063 $ & $0.063 $\\
				&	 & & $ 0.024$  & $0.072 $ & $0.072  $ & $ 0.009 $ & $0.073 $ & $0.073 $ & $0.012  $ & $0.074 $ & $0.074$\\
				&	 & $\gamma_4 $ & $0.010$  & $0.045 $ & $ 0.045 $ & $0.016 $ & $ 0.054$ & $ 0.054$ & $ 0.018 $ & $0.061 $ & $0.061 $\\
				&	 & & $- 0.014$  & $0.049 $ & $0.049 $ & $ -0.019 $ & $0.058 $ & $0.059 $ & $-0.018  $ & $0.069 $ & $0.070$\\
				&	 & $\gamma_5 $ & $0.000 $  & $0.101 $ & $ 0.101 $ & $-0.019  $ & $ 0.097$ & $ 0.097$ & $ -0.012 $ & $0.107 $ & $0.108 $\\
				&	 & & $ 0.034$  & $0.111 $ & $0.112 $ & $ 0.033$ & $0.111 $ & $0.112 $ & $0.038 $ & $0.119 $ & $0.120 $\\
				&	 & $\gamma_6 $ & $-0.039 $  & $0.084 $ & $ 0.086$ & $-0.045 $ & $ 0.081$ & $ 0.084$ & $ -0.041 $ & $0.090$ & $0.092 $\\
				&	 & & $- 0.002$  & $0.094 $ & $0.094 $ & $ 0.015$ & $0.091 $ & $0.092 $ & $0.018  $ & $0.098 $ & $0.098 $\\
				&	 & $\beta_1 $ & $0.00 0$  & $0.004 $  & $0.004 $ & $0.003 $ & $0.009 $ & $0.009  $ & $-0.002 $ & $ 0.015$ & $0.015$\\
				&	 &  & $0.000 $  & $0.004 $ & $0.004 $ & $0.005  $ & $0.009 $ & $0.009$ & $ 0.001 $ & $ 0.015$ & $0.015$\\
				&	 & $\beta_2 $ & $-0.003$  & $0.009$  & $0.009  $ & $-0.007$ & $0.023$ & $0.023  $ & $-0.011 $ & $ 0.045$ & $0.045$\\
				&	 &  & $-0.002 $  & $0.009$ & $0.009  $ & $-0.004  $ & $0.024 $ & $0.024 $ & $ -0.006 $ & $ 0.047$ & $0.047 $\\
				&	 & $\beta_3 $ & $0.005 $  & $0.016 $  & $0.017 $ & $0.012 $ & $0.039$ & $0.039  $ & $0.015$ & $ 0.072$ & $0.072$\\
				&	 &  & $0.004 $  & $0.017 $ & $0.017  $ & $0.006  $ & $0.040$ & $0.040$ & $ 0.006 $ & $ 0.074$ & $0.074$\\
			\end{tabular}
	}}
\end{table}

Simulations show that the 2-step approach improves considerably upon \texttt{smcure} for estimation of $\gamma$ when $n=200$ and the censoring rate among the uncured observations is higher. In almost all scenarios the 2-step approach has a smaller variance, which is expected due to presmoothing, but it also often exhibits a lower bias. As the sample size increases or the censoring rate decreases, we see less difference between the two methods but still, the 2-step approach is usually better.  In terms of $\beta$ estimators, both approaches give very similar results. In addition, the boxplots in Figure~\ref{fig:non_conv}, indicate that, even when \texttt{smcure} does not converge, the 2-step approach still gives more reasonable estimates. 

Since the second step of the new method does not depend on the latency model, we expect it to be more stable than \texttt{smcure} with respect to misspecifications of the latency model. We investigate this issue by considering two additional settings: one corresponding to a non-Cox latency model (Model 5 below) for which we still apply the two methods as if the Cox model was true and one corresponding to a logistic-Cox model but in which we don't use the correct covariates. For the latter, we use Model 4, scenario 2 described above but fit a latency model with covariates $X_1,\dots, X_5$ instead of $Z_1,Z_2,Z_3$. In particular, this means that we are including covariates $X_1,X_2,X_5$ that actually do not have any effect and are excluding $Z_1$ which affects the survival of the uncured.

\textit{Model 5.} Both incidence and latency depend on three independent covariates $X_1=Z_1\sim \text{Unif}(-1,1)$, $X_2=Z_2\sim \text{Bernoulli}(0.4)$ and $X_3=Z_3\sim \text{Bernoulli}(0.6)$. The cure status  and the survival times for the uncured observations are generated as in Model 1 for $\gamma=(1.4,2,1,-1) $, $\rho$ depending on $z$, $\rho(z)=0.75+\exp(\beta^Tz)$, $\mu=1.5$, $\beta=(1,0.4,-0.6)$. In particular this means that the latency model does not satisfy the proportional hazards assumption. For an observation with covariate $z$, the event time is truncated at $\tau_0(z)$ equal to the $97\%$ quantile of the Weibull distribution with parameters $\rho(z)$ and $\mu^{-1/\rho(z)}$. The censoring times are  generated according to a Weibull proportional hazards model
\[
S_C(t|x)=\exp\left(-\frac{1}{22}\mu t^{\tilde\rho}\exp\left(\gamma^Tx\right)\right),
\] 
with $\mu=1.5$ and $\tilde{\rho}=2.5$, truncated at $\tau=\max_z \tau_0(z)+2$. This scenario corresponds to a cure rate of $30\%$ and a censoring rate of $45\%$. 

Results for sample size $200$ and $400$, reported in Table~\ref{tab:results_mis1}, show that when the true latency model is not a Cox proportional hazards model, even $\gamma$ estimates are biased. However, the 2-step approach has lower bias and MSE, hence suffers less from the misspecification of the latency. On the other hand, misspecification of the latency covariates when the model is still Cox, seems to be less critical. It leads to a slight increase in bias and variance compared to the results in Table~\ref{tab:results2}-\ref{tab:results3} but again the 2-step approach performs better. 

\begin{table}
	\caption{\label{tab:results_mis1}Bias, variance and MSE of $\hat\gamma$ and $\hat\beta$ for \texttt{smcure} (second rows) and the 2-step approach (first rows) in Model 5 and Model 4 with misspecification.}
	\centering
	%	\addtolength{\tabcolsep}{-2pt}
	\scalebox{0.85}{
		\fbox{
			\begin{tabular}{ccrrrrrr}
				%	\hline
				%			& & & & & & & & & && \\[-10pt]
				%	&&&\multicolumn{9}{c}{Censoring scheme}\\
				%	\cline{4-12}
				& & & & & & &  \\[-7pt]
				&&\multicolumn{3}{c}{$n=200$}&\multicolumn{3}{c}{$n=400$}\\
				Mod.&	 Par. &  Bias & Var. & MSE & Bias & Var. & MSE \\[2pt]
				\hline
				& & & & & & & \\[-8pt]
				5	  & $\gamma_1 $ & $-0.154 $  & $0.187 $ & $0.211  $ & $ -0.087$ & $0.100 $ & $0.107 $ \\
				&	 & $0.514$  & $0.393$ & $0.657  $ & $0.437 $ & $0.152$ & $ 0.343$ \\
				&  $\gamma_2 $ & $-0.304 $  & $0.277$ & $ 0.370 $ & $0.213 $ & $ 0.144$ & $ 0.189$ \\
				& &  $ 0.623$  & $0.464$ & $0.852  $ & $ 0.507$ & $0.184$ & $0.441 $\\
				&  $\gamma_3 $ & $-0.232$  & $0.203 $ & $ 0.257$ & $-0.166 $ & $ 0.108$ & $ 0.136$ \\
				& &  $ 0.297$  & $0.655$ & $0.744 $ & $ 0.235 $ & $0.176 $ & $0.231 $ \\
				&  $\gamma_4 $ & $0.207 $  & $0.188 $ & $ 0.231 $ & $0.132  $ & $ 0.112$ & $ 0.130$ \\
				& &  $ -0.326$  & $0.377 $ & $0.483  $ & $ -0.281$ & $0.166 $ & $0.245 $ \\
				& $\beta_1 $ & $-0.281 $  & $0.051 $  & $0.130$ & $-0.296$ & $0.025$ & $0.112  $ \\
				&  & $-0.354 $  & $0.056 $ & $0.181 $ & $-0.353  $ & $0.027 $ & $0.151$ \\
				&  $\beta_2 $ & $-0.137$  & $0.043
				$  & $0.062 $ & $-0.144 $ & $0.022 $ & $0.043  $\\
				&   & $-0.174 $  & $0.046$ & $0.076  $ & $-0.169 $ & $0.023 $ & $0.052$ \\
				&  $\beta_3 $ & $0.221 $  & $0.051$  & $0.099  $ & $0.210$ & $0.023 $ & $0.067 $ \\
				& &   $0.257$  & $0.054 $ & $0.120 $ & $0.235  $ & $0.024 $ & $0.079$ \\
				\cline{1-8}
				& & & & & & & \\[-8pt]
				$4$ & $\gamma_1 $ & $0.036 $  & $0.250 $ & $0.251 $ & $ 0.016 $ & $0.116 $ & $0.116$  \\
				&	  & $0.097 $  & $0.360 $ & $0.370  $ & $0.041  $ & $0.133$ & $ 0.135$ \\
				&	  $\gamma_2 $ & $0.059 $  & $0.066$ & $ 0.070$ & $0.050  $ & $ 0.036$ & $ 0.039$ \\
				&	  & $- 0.089$  & $0.092$ & $0.100 $ & $ -0.038$ & $0.038 $ & $0.040 $ \\
				&	 $\gamma_3 $ & $-0.029 $  & $0.150 $ & $ 0.151 $ & $-0.023  $ & $ 0.075$ & $ 0.075$\\
				&	  & $ 0.025$  & $0.222 $ & $0.223  $ & $ 0.008 $ & $0.091$ & $0.091 $\\
				&	  $\gamma_4 $ & $0.005$  & $0.109 $ & $ 0.109$ & $0.018 $ & $ 0.055$ & $ 0.055$\\
				&	  & $- 0.059$  & $0.149 $ & $0.152 $ & $ -0.018 $ & $0.061$ & $0.062 $\\
				&	  $\gamma_5 $ & $-0.057 $  & $0.204$ & $ 0.207 $ & $-0.022 $ & $ 0.101$ & $ 0.102$\\
				&	  & $ 0.022$  & $0.307 $ & $0.308 $ & $ 0.032$ & $0.120 $ & $0.121$\\
				&	  $\gamma_6 $ & $-0.029 $  & $0.204 $ & $ 0.205$ & $-0.047 $ & $ 0.103$ & $ 0.105$\\
				&	 & $ 0.075$  & $0.287 $ & $0.293 $ & $ 0.014$ & $0.122 $ & $0.122 $\\
				&	  $\beta_1 $ & $-0.025 $  & $0.028 $  & $0.029 $ & $-0.017 $ & $0.011 $ & $0.012  $\\
				&	   & $-0.004 $  & $0.030 $ & $0.030 $ & $-0.002  $ & $0.012 $ & $0.012$\\
				&	  $\beta_2 $ & $0.017$  & $0.068$  & $0.069  $ & $0.007$ & $0.032$ & $0.032  $\\
				&	   & $0.010$  & $0.074$ & $0.074  $ & $0.002  $ & $0.033 $ & $0.033 $ \\
				&	  $\beta_3 $ & $0.004 $  & $0.054 $  & $0.054  $ & $-0.001 $ & $0.024$ & $0.024$\\
				&	   & $0.011 $  & $0.058 $ & $0.059  $ & $0.002  $ & $0.025$ & $0.025$ \\
				&	  $\beta_4 $ & $0.013$  & $0.094$  & $0.095  $ & $0.012$ & $0.043$ & $0.043  $\\
				&	   & $0.004$  & $0.103$ & $0.103 $ & $0.005  $ & $0.044 $ & $0.044 $ \\
				&	  $\beta_5 $ & $-0.001 $  & $0.112 $  & $0.112 $ & $0.011 $ & $0.047$ & $0.047$\\
				&	   & $-0.019 $  & $0.119 $ & $0.119  $ & $0.000  $ & $0.049$ & $0.049$ \\
			\end{tabular}
	}}
\end{table}

\section{Application}
\label{sec:application}
In this section we illustrate the practical use of the method through two medical datasets for melanoma cancer patients and compare the results with those provided by the \texttt{smcure} package. Melanoma is a common skin cancer type for which nowadays it is expected that a considerable fraction of the patients get cured as a consequence of medical advances in diagnostics and treatment.  Therefore, it is important to account for the presence of cured patients in the statistical analysis of melanoma survival data and to evaluate new treatments focusing on cure and not only survival prolongation. 

\subsection{Eastern Cooperative Oncology Group (ECOG) Data}
The ECOG phase III clinical trial e1684 aimed at evaluating the effect of treatment (high dose interferon alpha-2b regimen) as the postoperative adjuvant therapy for melanoma patients. The corresponding dataset, consisting of $284$ observations (after deleting missing data), is available in the \texttt{smcure} package \cite{cai_smcure}. The event time is the time from initial treatment to recurrence of melanoma and three covariates have been considered: age (continuous variable centered to the mean), gender (0=male and 1=female) and treatment (0=control and 1=treatment). Around $30\%$ of the observations are censored. The Kaplan-Meier curve is shown in Figure~\ref{fig:KM_melanoma_1}. 

We fit a logistic-Cox mixture cure model by using the maximum likelihood principle (\texttt{smcure} package) and the proposed 2-step approach. For our method we use the \texttt{smcure} estimator as a preliminary estimator. In both cases, standard errors are computed through  $500$ {naive} bootstrap samples. 
The resulting parameter estimates, standard errors and corresponding p-values for the Wald test are reported  in Table~\ref{tab:melanoma1}. 

\begin{figure}
	\centering
	\makebox{
		\includegraphics[width=0.48\linewidth]{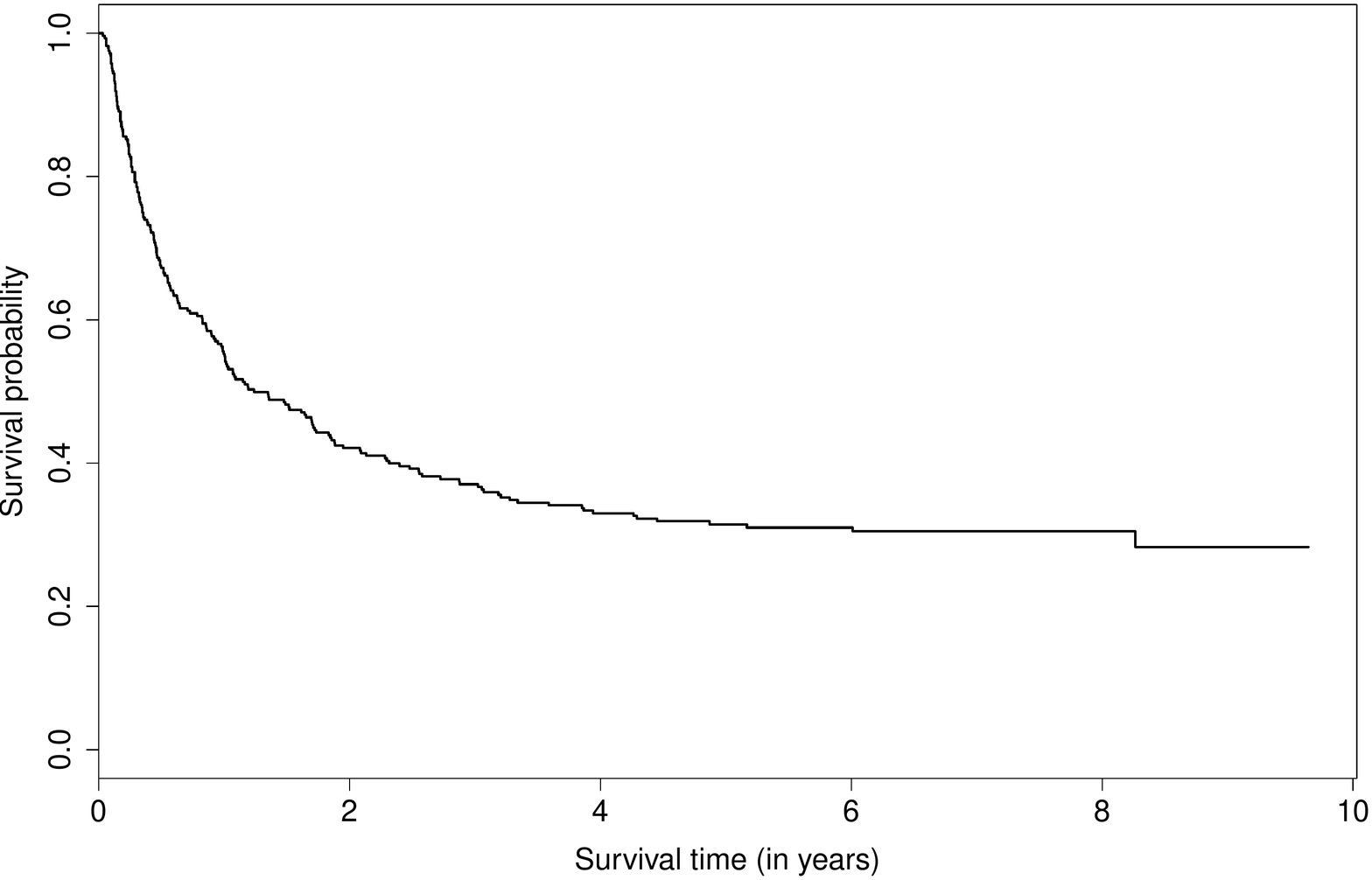}\qquad\includegraphics[width=0.48\linewidth]{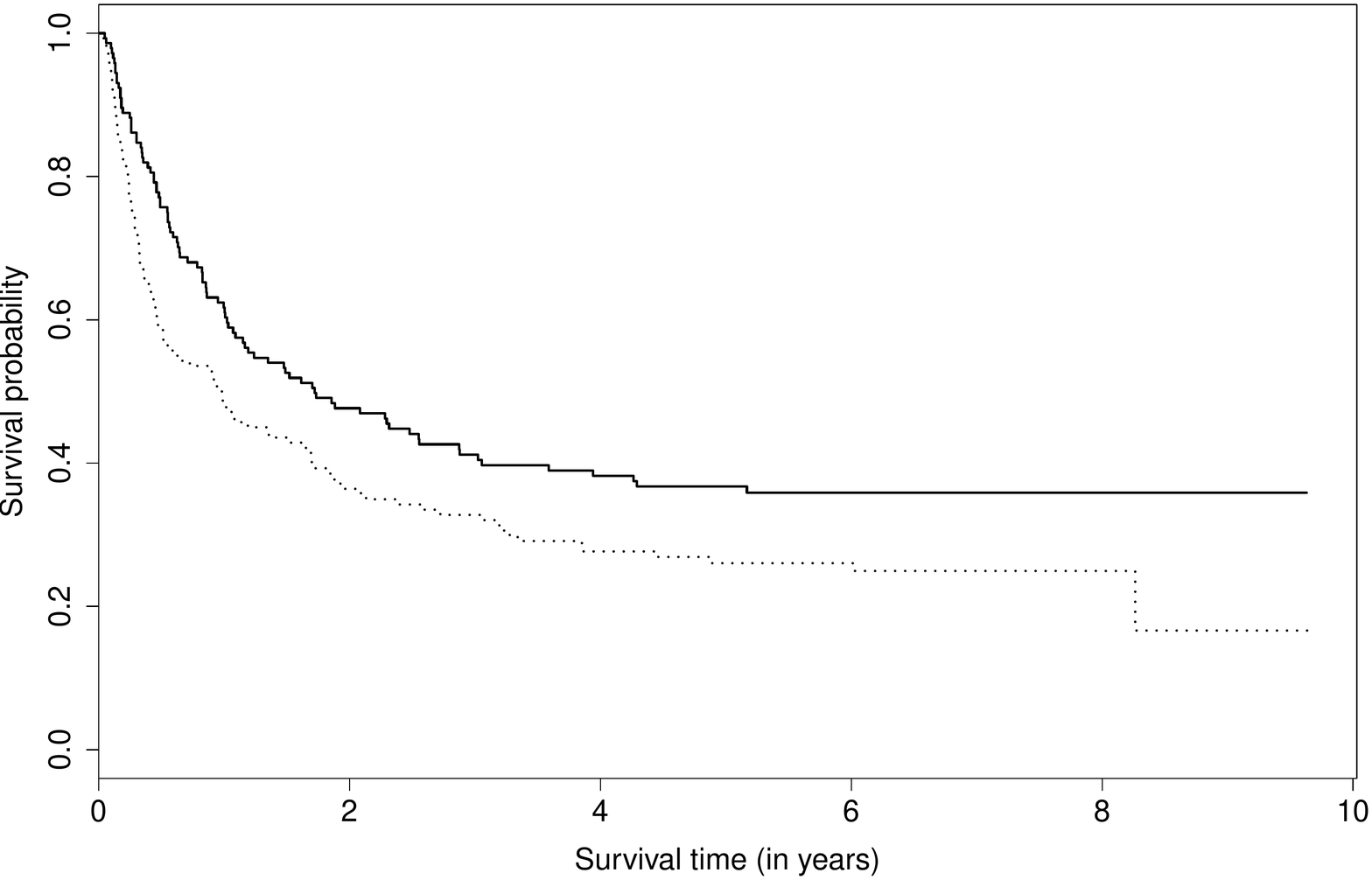}}
	\caption{\label{fig:KM_melanoma_1}Left panel: Kaplan-Meier survival curve for ECOG data. Right panel: Kaplan-Meier survival curves for the treatment group (solid) and control group (dotted) in the ECOG data.}
\end{figure}

\begin{table}
	\caption{\label{tab:melanoma1}Results for the incidence (logistic component) and the latency (Cox PH component) from the ECOG data.}
	\centering
	\scalebox{0.85}{
		\fbox{
			\begin{tabular}{c|crrrrrr}
				%	\hline
				&& \multicolumn{3}{c}{}  & \multicolumn{3}{c}{} \\[-8pt]
				&	& \multicolumn{3}{c}{\texttt{smcure} package}   & \multicolumn{3}{c}{2-step approach}\\
				&	Covariates	& Estimates & SE & p-value & Estimates & SE & p-value\\[2pt]
				\cline{2-8}
				&	& & & & & & \\[-8pt]
				\multirow{4}{*}{\STAB{\rotatebox[origin=c]{90}{incidence}}}	&Intercept & $1.3649 $ & $0.3457 $ & $8\cdot10^{-5} $   & $1.8351$ & $ 0.4924$ & $0.0002$\\
				&	Age & $0.0203 $ & $0.0159 $ & $0.2029 $   & $0.0388 $ & $0.0191 $ & $0.0418 $\\
				&	Gender & $-0.0869 $ & $0.3347 $ & $0.7949 $  & $-0.0864 $ & $0.3447 $ & $ 0.8126$\\
				&	Treatment & $-0.5884 $ & $ 0.3706$ & $0.1123 $   & $-1.1096$ & $0.5283$ & $0.0357 $\\
				\hline
				&	& & & & & &   \\[-8pt]
				\multirow{3}{*}{\STAB{\rotatebox[origin=c]{90}{latency}}}	&	Age & $-0.0077 $ & $0.0069 $ & $0.2663 $  &  $-0.0103$ & $0.0068 $ & $0.1319 $\\
				&	Gender & $0.0994 $ & $0.1932 $ & $0.6067 $   & $0.0629 $ & $0.1831 $ & $0.7313$\\
				&	Treatment & $-0.1535 $ & $0.1715 $ & $0.3707 $  & $-0.0526 $ & $ 0.1904$ & $0.7825 $\\
				%	\hline
	\end{tabular}}}
\end{table}

We observe that, despite exhibiting the same effect directions for all covariates, the two approaches give quite different results in terms of treatment effect. Age and treatment are both found to have a significant effect on the cure fraction when using the 2-step method, while \texttt{smcure} does not detect any significant effect. We also compare the two methods in terms of prediction accuracy for the incidence in the following way. As in \cite{AKL19}, we split the data into a training and a test set (at a 2:1 ratio), fit the model in the training set and then compute the prediction error for the test set according to the formula
\[
PE=-\sum_{j\in\text{test set}}\log \left[\phi\left(\hat\gamma_n^TX_j\right)^{\hat{w}_j}\{1-\phi\left(\hat\gamma_n^TX_j\right)\}^{1-\hat{w}_j}\right],
\]
where $\hat\gamma_n$ are the parameter estimates from the training set and $\hat{w}_j$ are the predicted uncure probabilities given the observations, i.e. 
\[
\hat{w}_j=\Delta_j+(1-\Delta_j)\frac{\phi(\hat\gamma_n^TX_j)\hat{S}_u(Y_j\mid  Z_j)}{1-\phi(\hat\gamma^T_nX_j)+\phi(\hat\gamma^T_nX_j)\hat{S}_u(Y_j\mid  Z_j)},
\]
We repeat this procedure 1000 times, for random selection of the train and test set. The boxplot of the difference between  the PE of the new method and the PE of \texttt{smcure}, over these 1000 iteration, is given in Figure~\ref{fig:boxplot_smcure}. We observe that the 2-step approach leads to lower PE (negative difference) in more that $50\%$ of the cases and the improvement in PE for the new method is usually larger compared to the cases in which \texttt{smcure} does better.

In addition, we expect the new approach to be more stable with respect to the latency model since that does not influence the second step of the estimation. To illustrate this point, we also fit a cure model with only gender as covariate for the survival of uncured patients (see Table~\ref{tab:melanoma3}) and see that in that case, \texttt{smcure} also detects the effect of the treatment to be significant. 

\begin{figure}
	\centering
	\makebox{
		\includegraphics[width=0.48\linewidth]{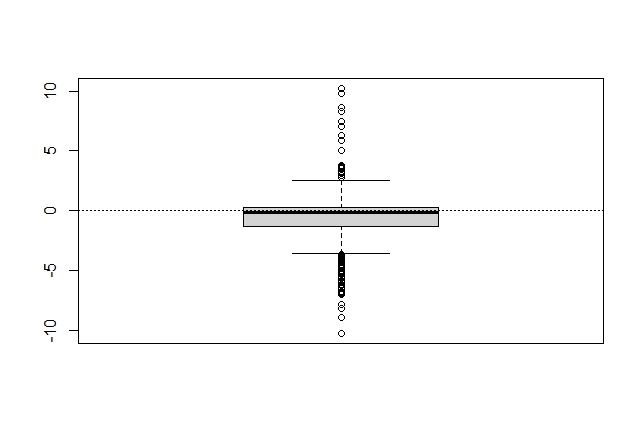}}
	\caption{\label{fig:boxplot_smcure}Boxplot of the difference between  the PE of the new method and the PE of \texttt{smcure}, over these 1000 iterations for the ECOG data.}
\end{figure}
\begin{table}
	\caption{\label{tab:melanoma3}Results for the incidence (logistic component) and the latency (Cox PH component) from the ECOG data.}
	\centering
	\scalebox{0.85}{
		\fbox{
			\begin{tabular}{c|crrrrrr}
				%	\hline
				&& \multicolumn{3}{c}{}  & \multicolumn{3}{c}{} \\[-8pt]
				&	& \multicolumn{3}{c}{\texttt{smcure} package}   & \multicolumn{3}{c}{2-step approach}\\
				&	Covariates	& Estimates & SE & p-value & Estimates & SE & p-value\\[2pt]
				\cline{2-8}
				&	& & & & & & \\[-8pt]
				\multirow{4}{*}{\STAB{\rotatebox[origin=c]{90}{incidence}}}	&Intercept & $1.4000 $ & $0.2791 $ & $5\cdot10^{-7} $   & $1.9073$ & $ 0.5225$ & $0.0002$\\
				&	Age & $0.0165$ & $0.0121 $ & $0.1709 $   & $0.0357$ & $0.0174 $ & $0.0399$\\
				&	Gender & $-0.0538 $ & $0.3101 $ & $0.8623$  & $-0.0429 $ & $0.3979$ & $ 0.9141$\\
				&	Treatment & $-0.6765$ & $ 0.3118$ & $0.0300$   & $-1.2590$ & $0.5809$ & $0.0302 $\\
				\hline
				&	& & & & & &   \\[-8pt]
				%	\multirow{1}{*}{\STAB{\rotatebox[origin=c]{90}{latency}}}
				latency	&	Gender & $0.0637 $ & $0.1935$ & $0.7421 $   & $0.0235$ & $0.1924 $ & $0.9028$\\
				%	\hline
	\end{tabular}}}
\end{table}

\subsection{Surveillance, Epidemiology and End Results database}

Here we consider melanoma data extracted from the SEER database to illustrate the performance of the method for more than one continuous covariate. The SEER database collects cancer incidence data from population-based cancer registries in US. We select the database `Incidence - SEER 18 Regs Research Data' and, in order to have a reasonable sample size, we extract the melanoma cancer data for the county of San Francisco in California during the period $2005-2010$. We consider only patients with known follow-up time  and tumor size (in the range 1-90 mm) and restrict the study to white people because of the very small number of cases from other races. The event of interest is death because of melanoma.  This cohort consists of $384$ observations out of which $228$ are male. The age ranges from $23$ to $101$ years old, the follow-up from $1$ to $143$ months with no events observed after $108$ months. Because of the high expected cure rate, $89\%$ of the observations are censored. We consider as covariates in the model: gender (0=male, 1=female), age and tumor size (continuous). The use of cure models is justified from the presence of a long plateau containing around $25\%$ of the observations (see the Kaplan-Meier curve in Figure \ref{fig:KM_melanoma_2}). 

\begin{figure}
	%	\makebox{
	\centering
	{\includegraphics[width=0.48\linewidth]{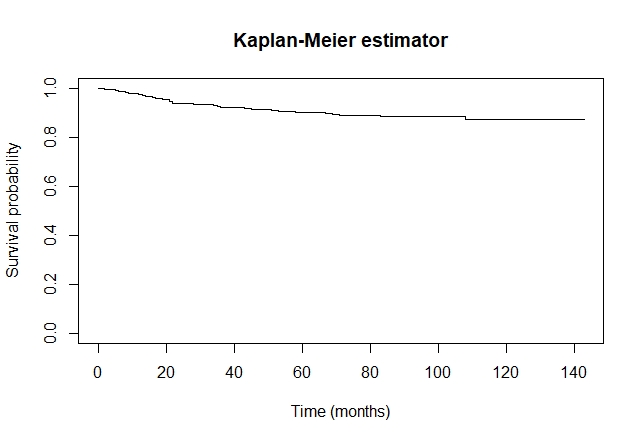}
	}
	%	}
	\caption{\label{fig:KM_melanoma_2}
		Kaplan-Meier survival curve for the SEER data. 
	}
\end{figure}

As in the previous section, we compute parameter estimates, standard errors and corresponding p-values for both methods (see Table \ref{tab:melanoma2}). We observe that both methods agree on the directions of the effects and give similar parameter estimates. Note that the sample size in this case is larger than in the previous  data example. However, \texttt{smcure} only finds age to be significant while the 2-step approach also detects the tumor size.  For the latency, none of the covariates is found significant with both methods. 
Also in this case, most of the time the new method leads to an improvement in terms of prediction errors, computed according to the procedure described in the previous subsection (see Figure~\ref{fig:boxplot_seer}). In addition, we also observe that if we remove the covariate tumor size from the latency model, the 2-step approach gives similar results while this time \texttt{smcure} also detects  tumor size as significant for the incidence component (see Table~\ref{tab:melanoma2_sub}). Once more, this behavior reflects the strong dependence of the incidence estimates on the latency model for \texttt{smcure}.  
\begin{table}
	\caption{\label{tab:melanoma2}
		Results for the incidence (logistic component) and the latency (Cox PH component) from the  SEER data.}
	\centering
	\scalebox{0.85}{
		\fbox{
			\begin{tabular}{c|crrrrrr}
				%	\hline
				&& \multicolumn{3}{c}{}  & \multicolumn{3}{c}{} \\[-8pt]
				&	& \multicolumn{3}{c}{\texttt{smcure} package}   & \multicolumn{3}{c}{2-step approach}\\
				&	Covariates	& Estimates & SE & p-value & Estimates & SE & p-value\\[2pt]
				\cline{2-8}
				&	& & & & & & \\[-8pt]
				\multirow{4}{*}{\STAB{\rotatebox[origin=c]{90}{incidence}}}&	Intercept & $-4.4952 $ & $0.9670$ & $3\cdot10^{-6} $  & $-4.9634$ & $1.1269 $ & $10^{-5}$ \\
				&	Age & $0.0411$ & $0.0142 $ & $0.0037 $  & $0.0455$ & $0.0167 $ & $ 0.0063$ \\
				&	Gender & $ -0.6805$ & $0.4551 $ & $0.1349 $  & $-0.7902 $ & $0.5162$ & $0.1259$ \\
				&Tumor size& $0.0214 $ & $0.0136 $ & $0.1147 $  & $ 0.0295$ & $0.0131 $ & $ 0.0239$ \\
				\hline
				&	& & & & & &  \\[-8pt]
				\multirow{3}{*}{\STAB{\rotatebox[origin=c]{90}{latency}}}&	Age & $-0.0161 $ & $ 0.0152$ & $ 0.2874$  & $-0.0171$ & $0.0151 $ & $0.2572$ \\
				&	Gender & $-0.3599 $ & $0.4217 $ & $ 0.3934$  & $-0.3125 $ & $0.4184 $ & $0.4551 $ \\
				&Tumor size& $0.0308 $ & $0.0224 $ & $0.1695$  & $0.0282 $ & $0.0217 $ & $0.1927 $ \\
				%	\hline
			\end{tabular}
	}}
\end{table}

\begin{figure}
	\centering
	\makebox{
		\includegraphics[width=0.48\linewidth]{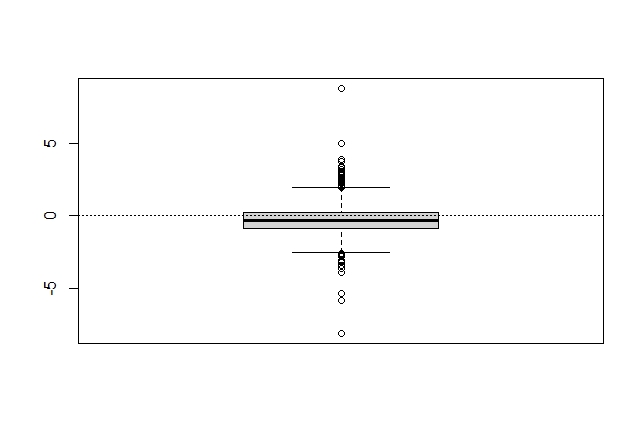}}
	\caption{\label{fig:boxplot_seer}Boxplot of the difference between  the PE of the new method and the PE of \texttt{smcure}, over these 1000 iterations for the SEER data.}
\end{figure}
\begin{table}
	\caption{\label{tab:melanoma2_sub}
		Results for the incidence (logistic component) and the latency (Cox PH component) from the  SEER data.}
	\centering
	\scalebox{0.85}{
		\fbox{
			\begin{tabular}{c|crrrrrr}
				%	\hline
				&& \multicolumn{3}{c}{}  & \multicolumn{3}{c}{} \\[-8pt]
				&	& \multicolumn{3}{c}{\texttt{smcure} package}   & \multicolumn{3}{c}{2-step approach}\\
				&	Covariates	& Estimates & SE & p-value & Estimates & SE & p-value\\[2pt]
				\cline{2-8}
				&	& & & & & & \\[-8pt]
				\multirow{4}{*}{\STAB{\rotatebox[origin=c]{90}{incidence}}}&	Intercept & $-4.7297 $ & $0.9213$ & $2\cdot10^{-7} $  & $-4.5477$ & $1.0829 $ & $2\cdot 10^{-5}$ \\
				&	Age & $0.0423$ & $0.0137$ & $0.0020 $  & $0.0392$ & $0.0159 $ & $ 0.0134$ \\
				&	Gender & $ -0.6393$ & $0.4579 $ & $0.1627$  & $-0.7021$ & $0.4430$ & $0.1130$ \\
				&Tumor size& $0.0288 $ & $0.0129 $ & $0.0251$  & $ 0.0250$ & $0.0129 $ & $ 0.0529$ \\
				\hline
				&	& & & & & &  \\[-8pt]
				latency&	Age & $-0.0165 $ & $ 0.0147$ & $ 0.2626$  & $-0.0145$ & $0.0146 $ & $0.3201$ \\
				&	Gender & $-0.3208 $ & $0.4282 $ & $ 0.4538$  & $-0.2559 $ & $0.4092 $ & $0.5316 $ \\
			\end{tabular}
	}}
\end{table}

\appendix
\section*{Appendix}
\section{Proofs}
\label{sec:proofs}
\begin{proof}[Proof of Theorem~\ref{theo:1}] We start by showing that, when $\tilde{\gamma}_n-\gamma_0=o_P(1)$, we have \begin{equation}
		\label{eqn:unif_cons}
		\sup_{x\in\tilde\X}\left|\hat\pi(x)-\pi_0(x)\right|=o_P(1).
	\end{equation}	
	{Note that, independently of the choice of the trimming function,  
		\[
		\E\left[\left\{\left(1-\pi_0(X)\right)\log{\phi(\gamma^TX)}+\pi_0(X)\log\left(1-\phi(\gamma^TX)\right)\right\}\tau(X)\right],
		\]	
		is maximized at $\gamma=\gamma_0$ because of condition (I1) and the fact that, for any $x\in\X$, the function 
		\[
		g_x(z)=\left\{\phi(\gamma_0^Tx)\log\frac{z}{\phi(\gamma_0^Tx)}+\left\{1-\phi(\gamma_0^Tx)\right\}\log\frac{1-z}{1-\phi(\gamma_0^Tx)}\right\}\tau(x),\quad z\in(0,1),
		\]
		is strictly negative for $z\neq \phi(\gamma_0^Tx)$ and  $g_x(\phi(\gamma_0^Tx))=0$. 
	}
	
	Next, we show uniform consistency of $\hat\pi_n$ using the decomposition \eqref{eqn:hat_pi_decomp} {and restricting to $\tilde{\X}$ as in \eqref{eqn:trimming}}. We have 
	\begin{equation}
		\label{eqn:sup_pi}
		\begin{aligned}
			\Vert\hat\pi_n-\pi_0\Vert_\infty&=\sup_{x\in\tilde\X}\left|\hat{g}_{\tilde{\gamma}_n}\left(\tilde{\gamma}^T_nx\right)-g_{\gamma_0}\left(\gamma^T_0x\right)\right|\\
			&=\sup_{x\in\tilde\X}\left|\hat{g}_{{\gamma}_0}\left({\gamma}^T_0x\right)-g_{\gamma_0}\left(\gamma^T_0x\right)\right|+\sup_{x\in\tilde\X}\left|\hat{g}_{\tilde{\gamma}_n}\left(\tilde{\gamma}^T_nx\right)-\hat{g}_{\gamma_0}\left(\gamma^T_0x\right)\right|\\
			&\leq \sup_{x\in\tilde\X}\sup_{t\leq\tau_0}\left|\hat{F}_n\left(t\mid \gamma^T_0X=\gamma^T_0x\right)-{F}_T\left(t\mid \gamma^T_0X=\gamma^T_0x\right)\right|\\
			&\qquad+\Vert\tilde{\gamma}^T_n-\gamma_0\Vert \sup_{x\in\tilde\X}\sup_{\gamma\in G}\left|\nabla_\gamma\hat{g}_{{\gamma}}\left({\gamma}^Tx\right)\right|.
		\end{aligned}
	\end{equation}
	The first term on the right-hand side of the equation converges to zero by Theorem 4.1 in \cite{KA99}. The second term converges to zero because of assumption (C1) and the fact that  
	\begin{equation}
		\label{eqn:hat_g}
		\sup_{x\in\tilde\X}\sup_{\gamma\in G}\left|\nabla_\gamma\hat{g}_{{\gamma}}\left({\gamma}^Tx\right)\right|=O_P(1).
	\end{equation} Indeed, \eqref{eqn:hat_g} follows from assumption (C9) and
	\begin{equation}
		\label{eqn:nabla_g}
		\sup_{x\in\tilde\X}\sup_{\gamma\in G}\left|\nabla_\gamma\hat{g}_{{\gamma}}\left({\gamma}^Tx\right)-\nabla_\gamma g_\gamma\left({\gamma}^Tx\right)\right|=o_P(1),
	\end{equation}
	which can be proved as in Lemma A.2 in \cite{lopez2013single}. Their estimator $\hat{G}_\theta(t|\lambda(\theta,x))$ is the same as our $\hat{F}_n(t|\gamma^TX=\gamma^Tx)$ if we replace $\theta$ by $\gamma$, consider $\lambda(\theta,x)=\theta^Tx$ and exchange $T$ with $C$ (they are interested in the conditional distribution of $C$ given $\lambda(\theta,X)$). This concludes the proof of \eqref{eqn:unif_cons}.
	 
	The consistency of $\hat\gamma_n$ then follows similarly to Theorem 1 in \cite{musta2020presmoothing}. There it is required that $\sup_{x\in\tilde\X}\left|\hat\pi(x)-\pi_0(x)\right|\to 0$ almost surely in order to obtain strong consistency of $\hat\gamma_n$. Here we restrict to convergence in probability and conclude via Theorem 1 in \cite{chen2003} formulated for M-estimators. 
\end{proof}
\begin{proof}[Proof of Theorem~\ref{theo:2}]
	The result follows from Theorem 2 in \cite{musta2020presmoothing} once we show that the assumptions (AN1)-(AN4) of that paper hold. {Note that the introduction of the trimming function $\tau(\cdot)$ would not change anything in the proof. It just allows us to restrict ourselves to the set $\tilde{\X}$ where the density of the index is bounded from below, in order to apply standard results from the literature.} The assumptions (AN1) and (AN3) of \cite{musta2020presmoothing} are the same as assumptions (C2), (N1) and (N2) here. It remains to verify assumptions (AN2) and (AN4) which for completeness we state below: 
	\begin{itemize}
		\item[(AN2)] $\pi_0(\cdot)$ belongs to a class of functions $\Pi$ such that 
		\begin{equation*}
			\label{eqn:entropy}
			\int_0^\infty \sqrt{\log N(\epsilon,\Pi,\Vert\cdot\Vert_{\infty})}\,\dd\epsilon<\infty,
		\end{equation*}
		where $N(\epsilon,\Pi,\Vert\cdot\Vert_{\infty})$ denotes the $\epsilon$-covering number of the space $\Pi$ with respect to $\Vert\pi\Vert_{{\infty}}=\sup_{x\in\tilde\X}|\pi(x)|$.
		\item[(AN4)] The estimator $\hat\pi(\cdot)$ satisfies the following properties:
		\begin{itemize}
			\item[(i)] $\p\left(\hat\pi(\cdot)\in\Pi\right)\to 1$.
			\item[(ii)] $\left\Vert\hat\pi(x)-\pi_0(x)\right\Vert_{{\infty}}=o_P(n^{-1/4})$. 
			\item[(iii)] There exists a function $\Psi$ such that 
			\[
			\begin{split}
				&\E^*\left[\left(\hat\pi(X)-\pi_0(X)\right)\left(\frac{1}{\phi(\gamma_0^TX)}+\frac{1}{1-\phi(\gamma_0^TX)}\right)\phi'(\gamma_0^TX)X{\tau(X)}\right]\qquad\qquad\\
				&\qquad=\frac{1}{n}\sum_{i=1}^n\Psi(Y_i,\Delta_i,X_i,Z_i)+R_n,
			\end{split}
			\]
			where $\E^*$ denotes the conditional expectation given the sample, taken with respect to the generic variable $X$, $\E[\Psi(Y,\Delta,X,Z)]=0$ and $\Vert R_n\Vert=o_P(n^{-1/2})$.
		\end{itemize}
	\end{itemize}
	
	\textit{Step 1.} For (AN2) we use the class of functions 
	\[
	\Pi=\left\{f:\X\to[0,1], \, f(x)=g(\gamma^Tx)\text{ for some } \gamma\in G \text{ and } g\in\tilde{\Pi}_\gamma\right\},
	\]
	where $\tilde{\Pi}_\gamma$ is the space of  continuously differentiable functions $g$ from $\X_\gamma$ 	to $[0,1]$ such that $\sup_{u\in\X_\gamma}|g'(u)|\leq M$ and
	\[
	\sup_{\substack{u_1\neq  u_2\\ u_1,u_2\in\X_\gamma}}\frac{|g'(u_1)-g'(u_2)|}{|u_1-u_2|^\xi}\leq M,
	\]
	for some $M>0$ and $\xi\in(0,1]$ independent of $\gamma$. The norm that we consider on $\Pi$ is the sup norm. By assumption (A2) we have $\pi_0\in\Pi$. Next we compute the $\epsilon$-covering number of $\Pi$. Let $\gamma_1,\dots,\gamma_{N_1}$ be a $\delta$-covering of the compact $G\in\R^p$ with respect to the $l_1$ norm. We have $N_1\leq \frac{K_1}{\delta^p}$ for some constant $K_1>0$. Consider the class $\tilde{\Pi}$ of continuously differentiable functions $g$ from $\mathcal{U}=\{\gamma^Tx\mid \gamma\in {G}, x\in\X\}$
	to $[0,1]$ such that $\sup_{u\in\mathcal{U}}|g'(u)|\leq M$ and
	\[
	\sup_{\substack{u_1\neq  u_2\\ u_1,u_2\in\mathcal{U}}}\frac{|g'(u_1)-g'(u_2)|}{|u_1-u_2|^\xi}\leq M.
	\]
	We have in particular that if $g\in\tilde{\Pi}$, then the restriction of $g$ to $\X_\gamma$ belongs to the class $\tilde{\Pi}_\gamma$ for any $\gamma\in\G$. Let $g_1,\dots,g_{N_2}$ be a $\rho$-covering of $\tilde{\Pi}$ with respect to the sup norm. Since $\mathcal{U}$ is bounded and convex ($G$ can be chosen to be a small neighborhood of $\gamma_0$), from Theorem 2.7.1 in \cite{VW96}, we have 
	\[
	\log N_2\leq K_2\left(\frac1\rho\right)^{1/(1+\xi)},
	\]
	for some constant $K_2>0$. For $1\leq i\leq N_1$ and $1\leq j\leq N_2$, define $f_{i,j}(x)=g_j(\gamma^T_ix)$. We show that $\{f_{i,j}\}$, $1\leq i\leq N_1$ and $1\leq j\leq N_2$ is an $\epsilon$-covering of $\Pi$. Let $f\in\Pi$, $f(x)=g(\gamma^Tx)$. From Whitney's theorem it follows that $g\in\tilde{\Pi}_\gamma$ can be extended to a function $\bar{g}\in\tilde{\Pi}$, i.e. $\bar{g}(u)=g(u)$ for $u\in\X_\gamma$. Let $i$ and $j$ be such that $\sup_{u\in\mathcal{U}}|\bar{g}(u)-g_j(u)|\leq \rho$  and $\Vert\gamma-\gamma_i\Vert\leq \delta. $ We then have
	\[
	\begin{aligned}
		\Vert f-f_{i,j}\Vert_\infty&\leq \sup_{x\in\X}\left|g\left(\gamma^Tx\right)-g_j\left(\gamma^Tx\right)\right|+\sup_{x\in\X}\left|g_j\left(\gamma^Tx\right)-g_j\left(\gamma^T_ix\right)\right|\\
		&\leq\sup_{u\in\mathcal{U}}|\bar{g}(u)-g_j(u)|+K_3\Vert\gamma-\gamma_i\Vert_1\\
		&\leq \rho+MK_3\delta.
	\end{aligned}
	\]
	Hence, if we take $\rho=\epsilon/2$ and $\delta=\epsilon/MK_3$ then we get an $\epsilon$-covering of $\Pi$. It follows that 
	\[
	\log N(\epsilon,\Pi,\Vert\cdot\Vert_\infty)\leq c_1-p\log\epsilon + c_2\left(\frac1\epsilon\right)^{1/(1+\xi)} \qquad c_1,c_2>0,
	\]
	and as a result (AN2) is satisfied. 
	
	\textit{Step 2.}	For assumption (AN4))(i), note that $\hat\pi_n(x)=\hat{g}_{\tilde{\gamma}_n}(\tilde{\gamma}^T_nx)$ with $\hat{g}_\gamma$ as in \eqref{def:hat_g}.  From consistency of $\tilde{\gamma}_n$ it follows that $\tilde{\gamma}_n\in G$. Let ${\hat{g}}'_{\tilde{\gamma}_n}$ denote the derivative of the function $u\mapsto \hat{g}_{\tilde{\gamma}_n}(u)$. Let $g_\gamma$ be as in \eqref{def:g}.   Next we show that 
	\begin{equation}
		\label{eqn:cond1}
		\sup_{x\in{\tilde\X}}|\hat{g}'_{\tilde{\gamma}_n}(\tilde{\gamma}^T_nx)-g'_{\gamma_0}(\gamma^T_0x)|=o_P(1),
	\end{equation}
	and
	\begin{equation}
		\label{eqn:cond2}
		\sup_{\substack{x_1,x_2\in{\tilde\X}\\
				\tilde{\gamma}^T_n	x_1\neq  \tilde{\gamma}^T_nx_2}}\frac{\left|\hat{g}'_{\tilde{\gamma}_n}\left(\tilde{\gamma}^T_nx_1\right)-g'_{\gamma_0}\left(\gamma^T_0x_1\right)-\hat{g}'_{\tilde{\gamma}_n}\left(\tilde{\gamma}^T_nx_2\right)+g'_{\gamma_0}\left(\gamma^T_0x_2\right)\right|}{|\tilde{\gamma}^T_nx_1-\tilde{\gamma}^T_nx_2|^\xi}=o_P(1),
	\end{equation}
	from which we can derive that (AN4)(i) is satisfied since $\sup_{x\in{\tilde\X}}|g'_{\gamma_0}(\gamma^T_0x)|<\infty$ and
	\[
	\sup_{\substack{x_1,x_2\in{\tilde\X}\\
			{\gamma}^T_0	x_1\neq  {\gamma}^T_0x_2}}\frac{\left|g'_{\gamma_0}\left(\gamma^T_0x_2\right)-g'_{\gamma_0}\left(\gamma^T_0x_1\right)\right|}{|{\gamma}^T_0x_1-{\gamma}^T_0x_2|^\xi}<\infty.
	\]
	From Theorems 4.1 and 4.2 in \cite{KA99} we have 
	\begin{equation}
		\label{eqn:KA1}	\sup_{x\in{\tilde\X}}|\hat{g}'_{{\gamma}_0}({\gamma}^T_0x)-g'_{\gamma_0}(\gamma^T_0x)|=o_P(1),
	\end{equation}
	and
	\begin{equation}
		\label{eqn:KA2}
		\sup_{\substack{x_1,x_2\in{\tilde\X}\\
				{\gamma}^T_0	x_1\neq  {\gamma}^T_0x_2}}\frac{\left|\hat{g}'_{{\gamma}_0}\left({\gamma}^T_0x_1\right)-g'_{\gamma_0}\left(\gamma^T_0x_1\right)-\hat{g}'_{{\gamma}_0}\left({\gamma}^T_0x_2\right)+g'_{\gamma_0}\left(\gamma^T_0x_2\right)\right|}{|{\gamma}^T_0x_1-{\gamma}^T_0x_2|^\xi}=o_P(1),
	\end{equation}
	by conditioning on the variable $\gamma_0^TX$. Next, we need to deal with the fact that we are using $\tilde{\gamma}_n^TX$ instead of $\gamma_0^TX$.
	
	We can write
	\[
	\log \hat{g}_{\tilde{\gamma}_n}(u)=\sum_{i=1}^n \Delta_i\log\left(1-W_i(u;\tilde{\gamma}_n)\right),
	\]
	where
	\[
	W_i(u;\tilde{\gamma}_n)=k\left(\frac{\tilde{\gamma}^T_nX_i-u}{b}\right)\bigg/\left(\sum_{j=1}^n\1_{\{T_j\geq T_i\}}k\left(\frac{\tilde{\gamma}^T_nX_j-u}{b}\right)\right).
	\]
	Hence 
	\begin{equation}
		\label{eqn:g'}
		\hat{g}'_{\tilde{\gamma}_n}(u)=\hat{g}_{\tilde{\gamma}_n}(u)\sum_{i=1}^n \Delta_i\frac{W'_i(u;\tilde{\gamma}_n)}{1-W_i(u;\tilde{\gamma}_n)}. 
	\end{equation}
	Using the triangular inequality we obtain
	\begin{equation}
		\label{eqn:der_g_hat}
		\begin{aligned}
			&\sup_{x\in{\tilde\X}}\left|\hat{g}'_{\tilde{\gamma}_n}(\tilde{\gamma}_n^Tx)-\hat{g}'_{{\gamma}_0}({\gamma}_0^Tx)\right|\\
			&\leq \sup_{x\in{\tilde\X}}\left|\hat{g}_{\tilde{\gamma}_n}(\tilde{\gamma}_n^Tx)-\hat{g}_{{\gamma}_0}({\gamma}_0^Tx)\right|\sup_{x\in{\tilde\X}}\left|\sum_{i=1}^n \Delta_i\frac{W'_i(\tilde{\gamma}_n^Tx;\tilde{\gamma}_n)}{1-W_i(\tilde{\gamma}_n^Tx;\tilde{\gamma}_n)}\right|\\
			&+\sup_{x\in{\tilde\X}}|\hat{g}'_{{\gamma}_0}({\gamma}_0^Tx)|\sup_{x\in\X}\left|\sum_{i=1}^n \Delta_i\left\{\frac{W'_i(\tilde{\gamma}_n^Tx;\tilde{\gamma}_n)}{1-W_i(\tilde{\gamma}_n^Tx;\tilde{\gamma}_n)}-\frac{W'_i({\gamma}_0^Tx;{\gamma}_0)}{1-W_i({\gamma}_0^Tx;{\gamma}_0)}\right\}\right|.
		\end{aligned}
	\end{equation}
	As is \eqref{eqn:sup_pi}, we have 
	\[
	\sup_{x\in{\tilde\X}}\left|\hat{g}_{\tilde{\gamma}_n}(\tilde{\gamma}_n^Tx)-\hat{g}_{{\gamma}_0}({\gamma}_0^Tx)\right|=O_P(n^{-1/2}).
	\]		
	Moreover, $\sup_{x\in{\tilde\X}}|\hat{g}'_{{\gamma}_0}({\gamma}_0^Tx)|=O_P(1)$ and, as in the proof of Proposition 4.1 in \cite{KA99} it can be seen that 
	\[
	\sup_{x\in{\tilde\X}}\left|\sum_{i=1}^n \Delta_i\frac{W'_i({\gamma}_0^Tx;{\gamma}_0)}{1-W_i({\gamma}_0^Tx;{\gamma}_0)}\right|=O_P(1).
	\]
	and 
	\[
	\sup_{x\in{\tilde\X}}\left|\sum_{i=1}^n \Delta_i\nabla_\gamma\left(\frac{W'_i({\gamma}^Tx;{\gamma})}{1-W_i({\gamma}^Tx;{\gamma})}\right)\bigg|_{\gamma_0}\right|=O_P(1/b).
	\]
	From a Taylor expansion and $\tilde{\gamma}_n-\gamma_0=O_P(n^{-1/2})$, we conclude that the left hand side of \eqref{eqn:der_g_hat} converges to zero.  Together with \eqref{eqn:KA1} we obtain \eqref{eqn:cond1}.
	
	Next we show that 
	\[
	\sup_{\substack{x_1,x_2\in{\tilde\X}\\
			\tilde{\gamma}^T_n	x_1\neq  \tilde{\gamma}^T_nx_2}}\frac{\left|\hat{g}'_{\tilde{\gamma}_n}\left(\tilde{\gamma}^T_nx_1\right)-\hat{g}'_{\gamma_0}\left(\gamma^T_0x_1\right)-\hat{g}'_{\tilde{\gamma}_n}\left(\tilde{\gamma}^T_nx_2\right)+\hat{g}'_{\gamma_0}\left(\gamma^T_0x_2\right)\right|}{|\tilde{\gamma}^T_nx_1-\tilde{\gamma}^T_nx_2|^\xi}=o_P(1).
	\]
	Note that by \eqref{eqn:g'} and the fact that $\sup_{x\in\X}|\hat{g}'_{{\gamma}}({\gamma}^Tx)|=O_P(1)$ for  both $\gamma=\gamma_0$ and $\gamma=\tilde{\gamma}_n$, it is sufficient to consider 
	\begin{equation}
		\label{eqn:V}
		\sup_{\substack{x_1,x_2\in{\tilde\X}\\
				\tilde{\gamma}^T_n	x_1\neq  \tilde{\gamma}^T_nx_2}}\frac{\left|V_{\tilde{\gamma}_n}\left(\tilde{\gamma}^T_nx_1\right)-V_{\gamma_0}\left(\gamma^T_0x_1\right)-V_{\tilde{\gamma}_n}\left(\tilde{\gamma}^T_nx_2\right)+V_{\gamma_0}\left(\gamma^T_0x_2\right)\right|}{|\tilde{\gamma}^T_nx_1-\tilde{\gamma}^T_nx_2|^\xi},
	\end{equation}
	where
	\[V_\gamma(u)=\sum_{i=1}^n \Delta_i\frac{W'_i(u;{\gamma})}{1-W_i(u;{\gamma})}.
	\]
	In addition we can also restrict the supremum over $x_1,x_2$ such that $\tilde{\gamma}^T_nx_1\neq \tilde{\gamma}^T_nx_2$ and ${\gamma}^T_0x_1\neq {\gamma}^T_0x_2$ because otherwise we would have $|\tilde{\gamma}^T_nx_1- \tilde{\gamma}^T_nx_2|\leq cn^{-1/2}$ and \eqref{eqn:V} would obviously hold for $\xi$ sufficiently small by using a Taylor expansion and $\sup_{u}V'_\gamma(u)=O_P(1/b)$.  Note also that \[
	\sup_{\substack{x_1,x_2\in{\tilde\X}\\
			\tilde{\gamma}^T_n	x_1\neq  \tilde{\gamma}^T_nx_2,\, {\gamma}^T_0	x_1\neq  {\gamma}^T_0x_2}}\frac{|{\gamma}^T_0x_1-{\gamma}^T_0x_2|}{|\tilde{\gamma}^T_nx_1-\tilde{\gamma}^T_nx_2|}=O_P(1),
	\] 
	hence the denominator can be replaced by $|{\gamma}^T_0x_1-{\gamma}^T_0x_2|^\xi$. 
	By Taylor expansion we obtain that the largest order term in \eqref{eqn:V} is 
	\[
	(\tilde{\gamma}_n-\gamma_0)	\sup_{\substack{x_1,x_2\in{\tilde\X}\\
			{\gamma}^T_0	x_1\neq  {\gamma}^T_0x_2}}\sum_{i=1}^n \Delta_i\frac{\nabla_\gamma\left(\frac{W'_i({\gamma}^Tx_1;{\gamma})}{1-W_i({\gamma}^Tx_1;{\gamma})}\right)\bigg|_{\gamma_0}-\nabla_\gamma\left(\frac{W'_i({\gamma}^Tx_2;{\gamma})}{1-W_i({\gamma}^Tx_2;{\gamma})}\right)\bigg|_{\gamma_0}}{|{\gamma}^T_0x_1-{\gamma}^T_0x_2|^\xi}.
	\] 
	Using Lemma A.1. in \cite{KA99} and that for a smooth function $l$ (in our case $l$ is the kernel function or its derivatives) we have $$|l(\frac{x-z}{b})-l(\frac{x-z}{b})|/|x-y|^\xi=O(b^{-\xi}), $$ uniformly over $x,y,z$, 
	we obtain that the expression in the previous equation is of the order $O_P(n^{-1/2}b^{-1-\xi})=o_P(1)$. Together with \eqref{eqn:KA2} this yields \eqref{eqn:cond2}.
	
	\textit{Step 3.}	Assumption (AN4)(ii) can be checked using  \eqref{eqn:sup_pi}. The first term on the right-hand side of that equation is of order $o_P(n^{-1/4})$ because of Theorem 4.1 in \cite{KA99} and the assumption (C6). The second term on the right-hand side of \eqref{eqn:sup_pi} is also of   order $o_P(n^{-1/4})$ because of assumptions (N3), (C9) and \eqref{eqn:nabla_g}.
	
	\textit{Step 4.}	For (AN4)(iii), using the decomposition \eqref{eqn:hat_pi_decomp} we have 
	\begin{equation}
		\begin{aligned}
			\label{eqn:expectation*}
			&\E^*\left[\left(\hat\pi(X)-\pi_0(X)\right)\left(\frac{1}{\phi(\gamma_0^TX)}+\frac{1}{1-\phi(\gamma_0^TX)}\right)\phi'(\gamma_0^TX)X{\tau(X)}\right]\\
			&=\E^*\left[\frac{\hat{F}_n(\tau_0\mid \gamma^T_0X)-{F}_T(\tau_0\mid \gamma^T_0X)}{\phi(\gamma_0^TX)\left\{1-\phi(\gamma_0^TX)\right\}}\phi'(\gamma_0^TX)X{\tau(X)}\right]\\
			&\quad+\E^*\left[\frac{\hat{g}_{\tilde{\gamma}_n}(\tilde{\gamma}^T_nX)-\hat{g}_{\gamma_0}(\gamma^T_0X)}{\phi(\gamma_0^TX)\left\{1-\phi(\gamma_0^TX)\right\}}\phi'(\gamma_0^TX)X{\tau(X)}\right].
		\end{aligned}
	\end{equation}
	Since $\sup_{x\in\X}H([t,\infty)|\gamma^T_0x)<1$ by condition  \eqref{eqn:cond_support}, from Theorem 3.2 in \cite{DA2002}, it follows as in \cite{musta2020presmoothing} (see proof of Theorem 5) that the first term on the right-hand side of \eqref{eqn:expectation*} is equal to 
	$
	\frac{1}{n}\sum_{i=1}^n \psi(Y_i,\Delta_i,X_i)+R_n
	$
	with $\Vert R_n\Vert=o_P(n^{-1/2})$ and 
	\begin{equation}
		\label{eqn:iid_pi}
		\begin{split}
			\psi(Y,\Delta,X)	=-
			\left\{\frac{\Delta\1_{\{Y\leq\tau_0\}}}{H\left([Y,\infty)|\gamma^T_0X\right)}-\int_0^{Y\wedge \tau_0}\frac{H_1\left(ds|\gamma^T_0X\right)}{H^2\left([s,\infty)|\gamma^T_0X\right)}\right\}\frac{\phi'\left(\gamma_0^TX\right)}{\phi\left(\gamma_0^TX\right)}X{\tau(X)}.
		\end{split}
	\end{equation}	
	For the second term in \eqref{eqn:expectation*}, by the mean value theorem, we write
	\[
	\begin{aligned}
		\hat{g}_{\tilde{\gamma}_n}\left(\tilde{\gamma}^T_nx\right)-\hat{g}_{\gamma_0}\left(\gamma^T_0x\right)
		&=(\tilde{\gamma}_n-\gamma_0)^T\nabla_\gamma\hat{g}_\gamma\left(\gamma^Tx\right)\big\vert_{\gamma^*}\\
		&=(\tilde{\gamma}_n-\gamma_0)^T\nabla_\gamma\hat{g}_\gamma\left(\gamma^Tx\right)\big\vert_{\gamma^*},
	\end{aligned}
	\]
	for some $\gamma^*$ such that $\Vert\gamma^*-\gamma_0\Vert\leq \Vert\tilde\gamma_n-\gamma_0\Vert$. Using \eqref{eqn:nabla_g}
	and assumption (C9),  we obtain 
	\[
	\hat{g}_{\tilde{\gamma}_n}\left(\tilde{\gamma}^T_nx\right)-\hat{g}_{\gamma_0}\left(\gamma^T_0x\right)=(\tilde{\gamma}_n-\gamma_0)^T\nabla_\gamma{g}_\gamma\left(\gamma^Tx\right)\big\vert_{\gamma_0}+o_P(n^{-1/2}),
	\]
	where the $o_P(n^{-1/2})$ term is uniform with respect to $x\in{\tilde\X}$. 
	It follows that 
	\[
	\begin{aligned}
		&\E^*\left[\frac{\hat{g}_{\tilde{\gamma}_n}(\tilde{\gamma}^T_nX)-\hat{g}_{\gamma_0}(\gamma^T_0X)}{\phi(\gamma_0^TX)\left\{1-\phi(\gamma_0^TX)\right\}}\phi'(\gamma_0^TX)X{\tau(X)}\right]\\
		&=\E\left[\frac{\nabla_\gamma{g}_\gamma(\gamma^TX)^T\big\vert_{\gamma_0}}{\phi(\gamma_0^TX)\left\{1-\phi(\gamma_0^TX)\right\}}\phi'(\gamma_0^TX)X{\tau(X)}\right](\tilde{\gamma}_n-\gamma_0)+o_P(n^{-1/2}),
	\end{aligned}
	\]
because $E^*$ is just the expectation with respect to the variable $X$.
	By definition of $g_\gamma(\gamma^Tx)$ and a Taylor expansion we have 
	\[
	\begin{aligned}
		g_\gamma\left(\gamma^Tx\right)&=\E\left[1-\phi\left(\gamma_0^TX\right)|\gamma^TX=\gamma^Tx\right]\\
		&=1-\phi\left(\gamma^Tx\right)+(\gamma-\gamma_0)^T\E\left[X|\gamma^TX=\gamma^Tx\right]\phi'\left(\gamma^Tx\right)\\
		&\quad+\frac12(\gamma-\gamma_0)^T\E\left[\phi''\left(\gamma^T_*x\right)XX^T|\gamma^TX=\gamma^Tx\right](\gamma-\gamma_0),
	\end{aligned}
	\]
	for some $\Vert\gamma_*-\gamma_0\Vert\leq \Vert\gamma-\gamma_0\Vert$.
	Hence 
	\[
	\nabla_\gamma{g}_\gamma\left(\gamma^Tx\right)\big\vert_{\gamma_0}=-\phi'\left(\gamma^T_0x\right)x+\E\left[X|\gamma^T_0X=\gamma^T_0x\right]\phi'\left(\gamma^T_0x\right).
	\]
	and consequently
	\[
	\begin{aligned}
		&\E\left[\frac{\nabla_\gamma{g}_\gamma(\gamma^TX)^T\big\vert_{\gamma_0}}{\phi(\gamma_0^TX)\left\{1-\phi(\gamma_0^TX)\right\}}\phi'(\gamma_0^TX)X{\tau(X)}\right]\\
		&=-\E\left[\frac{\phi'(\gamma^T_0X)X^T}{\phi(\gamma_0^TX)\left\{1-\phi(\gamma_0^TX)\right\}}\phi'(\gamma_0^TX)X{\tau(X)}\right]\\
		&\quad +\E\left[\frac{\E[X^T|\gamma^T_0X]\phi'(\gamma^T_0X)}{\phi(\gamma_0^TX)\left\{1-\phi(\gamma_0^TX)\right\}}\phi'(\gamma_0^TX)X{\tau(X)}\right]\\
		&=\E\left[\frac{\phi'(\gamma^T_0X)^2\left\{\E[X^T|\gamma_0^TX]\E[X\tau(X)|\gamma_0^TX]-\E[X^TX{\tau(X)}|\gamma_0^TX]\right\}}{\phi(\gamma_0^TX)\left\{1-\phi(\gamma_0^TX)\right\}}\right],
	\end{aligned}
	\]
	 We denote this expression by $Q$. This yields
	\[
	\begin{aligned}
		&\E^*\left[\frac{\hat{g}_{\tilde{\gamma}_n}(\tilde{\gamma}^T_nX)-\hat{g}_{\gamma_0}(\gamma^T_0X)}{\phi(\gamma_0^TX)\left\{1-\phi(\gamma_0^TX)\right\}}\phi'(\gamma_0^TX)X{\tau(X)}\right]=Q(x)(\tilde{\gamma}_n-\gamma_0)+o_P(n^{-1/2}),
	\end{aligned}
	\]
	and by assumption (N3)
	\[
	\begin{aligned}
		&	\E^*\left[\left(\hat\pi(X)-\pi_0(X)\right)\left(\frac{1}{\phi(\gamma_0^TX)}+\frac{1}{1-\phi(\gamma_0^TX)}\right)\phi'(\gamma_0^TX)X{\tau(X)}\right]\\
		&=\frac{1}{n}\sum_{i=1}^n \Psi(Y_i,\Delta_i,X_i)+Q\sum_{i=1}^n \zeta(Y_i,\Delta_i,X_i,Z_i) +R_n,
	\end{aligned}
	\]
	with $\Vert R_n\Vert=o_P(n^{-1/2})$. This means that assumption (AN4)(iii) of \cite{musta2020presmoothing} is satisfied with $\varPsi(Y,\Delta,X,Z)=\psi(Y,\Delta,X,Z)+\zeta(Y,\Delta,X,Z)$.
	
	This concludes the verification of the assumptions of Theorem 3 in \cite{musta2020presmoothing}. It also follows that the covariance matrix is given by
	\begin{equation}
		\label{def:Sigma_gamma}
		\Sigma_\gamma=(\Gamma'_1\Gamma_1)^{-1}\Gamma'_1V\Gamma_1(\Gamma'_1\Gamma_1)^{-1}=\Gamma_1^{-1}V\Gamma_1^{-1},
	\end{equation}	
	where $V=Var(\Psi(Y,\Delta,X,Z))$ and 
	\[
	\Gamma_1=-\E\left[\left(\frac{1}{\phi(\gamma_0^TX)}+\frac{1}{1-\phi(\gamma_0^TX)}\right)\phi'(\gamma_0^TX)^2XX^T{\tau(X)}\right].
	\]
\end{proof}

 \begin{acks}[Acknowledgments]
I. Van Keilegom acknowledges financial support from the European Research Council (2016-2021, Horizon 2020 and grant agreement 694409).
V. Patilea gratefully acknowledges support from the Joint Research Initiative ‘Models and mathematical processing of very large data’ under the aegis of Risk Foundation, in partnership with MEDIAMETRIE and GENES, France.
For the simulations we used the Lisa cluster of the Dutch national Supercomputer. The data that were analysed in Section \ref{sec:simulations} have been provided by the SEER programme (www.seer.cancer.gov) research data (1973–2014), National Cancer Institute, Division of Cancer Control and Population Sciences, Surveillance Research Program.
 \end{acks}

%%%%%%%%%%%%%%%%%%%%%%%%%%%%%%%%%%%%%%%%%%%%%%
%% Supplementary Material, if any, should   %%
%% be provided in {supplement} environment  %%
%% with title and short description.        %%
%%%%%%%%%%%%%%%%%%%%%%%%%%%%%%%%%%%%%%%%%%%%%%
%\begin{supplement}
%\stitle{???}
%\sdescription{???.}
%\end{supplement}

%% if your bibliography is in bibtex format, uncomment commands:
%\bibliographystyle{imsart-nameyear} % Style BST file (imsart-number.bst or imsart-nameyear.bst)
\bibliography{cure_models}       % Bibliography file (usually '*.bib')

%% or include bibliography directly:
% \begin{thebibliography}{cure_models}
% \bibitem{b1}
% \end{thebibliography}

\end{document}